\documentclass[journal,10pt]{IEEEtran}
\pdfoutput=1
\usepackage{cite}
\usepackage{amsthm}
\usepackage{amsmath}
\usepackage{comment}
\usepackage{booktabs}
\usepackage{stfloats}
\usepackage{bm}
\usepackage{setspace}
\usepackage{graphicx}
\usepackage[usenames,dvipsnames]{color}
\usepackage{amsfonts,amssymb,amsmath}
\usepackage{subfigure}
\usepackage{enumerate}
\usepackage{cases}
\usepackage{array}
\usepackage{colortbl}
\usepackage{slashbox}
\usepackage{multirow}
\usepackage{arydshln}
\usepackage{soul, color, xcolor}
\usepackage{float}
\soulregister{\cite}7 % ע��\cite����
\soulregister{\citep}7 % ע��\citep����
\soulregister{\citet}7 % ע��\citet����
\soulregister{\ref}7 % ע��\ref����
\soulregister{\pageref}7 % ע��\pageref����
\newtheorem{thm}{Theorem}
\newtheorem{cor}{Corollary}

\theoremstyle{remark}
\newtheorem{defn}{Definition}

\newtheorem{exam}{Example}

\definecolor{Gray}{gray}{0.85}
\definecolor{mycyan}{cmyk}{.3,0,0,0}

\newcolumntype{a}{>{\columncolor{Gray}}c}%\centering\backslash
\newcolumntype{b}{>{\columncolor{white}}c}
\newcolumntype{d}{>{\columncolor{mycyan}}c}

\allowdisplaybreaks
\begin{document}
\title{Minimizing Age-upon-Decisions in Bufferless System: Service Scheduling and Decision Interval}
% Authors, for the paper (add full first names)
\author{Shutong~Chen, Tianci~Zhang, Zhengchuan~Chen,~\IEEEmembership{Member,~IEEE,} Yunquan~Dong,~\IEEEmembership{Member,~IEEE,} Min~Wang,~\IEEEmembership{Member,~IEEE,} Yunjian~Jia,~\IEEEmembership{Member,~IEEE,} and Tony~Q.~S.~Quek,~\IEEEmembership{Fellow,~IEEE}
\thanks{S. Chen, T. Zhang, Z. Chen, and Y. Jia are with the School of Microelectronics and Communication Engineering, Chongqing University, Chongqing 400044, China (Emails: \{cst, ztc, czc, yunjian\}@cqu.edu.cn).}
\thanks{Y. Dong is with School of Electronic and Information Engineering, Nanjing University of Information Science and Technology, Nanjing, China (Email: yunquandong@nuist.edu.cn).}
\thanks{M. Wang is with the School of Optoelectronics Engineering, Chongqing University of Posts and Telecommunications, Chongqing, China. (E-mail:~wangm@cqupt.edu.cn).}
\thanks{Tony Q. S. Quek is with the Pillar of Information Systems Technology and Design, Singapore University of Technology and Design, Singapore. (Email:
tonyquek@sutd.edu.sg).}

%\author{Zhengchuan~Chen,~Guido~C.~Ferrante,~Howard~H.~Yang, and~Tony~Q.~S.~Quek
%\thanks{}
%This work was supported in part by the SUTD-ZJU Research Collaboration under Grant SUTD-ZJU/RES/01/2014 and the MOE ARF Tier 2 under Grant MOE2015-T2-2-104.}
%\thanks{}
}
%\markboth{Journal of \LaTeX\ Class Files,~Vol.~11, No.~4, December~2012}%
%{Shell \MakeLowercase{\textit{et al.}}: Bare Demo of IEEEtran.cls for Journals}
%\maketitle
\maketitle

\begin{abstract}
The data freshness at decision epochs of time-sensitive applications, e.g., auto-driving vehicles and autonomous underwater robots, is jointly affected by the statistics of update process and decision process. This work considers an update-and-decision system with a Poisson-arrival bufferless queue, where updates are delivered and processed for making decisions with exponential or periodic intervals. We use age-upon-decisions (AuD) to characterize timeliness of updates at decision moments, and the missing probability to specify whether updates are useful for decision-making. Our theoretical analyses
1) present the average AuDs and the missing probabilities for bufferless systems with exponential or deterministic decision intervals under different service time distributions;
2) show that for service scheduling, the deterministic service time achieves a lower average AuD and a smaller missing probability than the uniformly distributed and the negative exponentially distributed service time;
3) prove that the average AuD of periodical decision system is larger than and will eventually drop to that of Poisson decision system along with the increase of decision rate; however, the missing probability in periodical decision system is smaller than that of Poisson decision system.
The numerical results and simulations verify the correctness of our analyses, and demonstrate that the bufferless systems outperform the systems applying infinite buffer size.
\end{abstract}


\section{Introduction}
\setlength{\textfloatsep}{0.9\baselineskip plus 0.2\baselineskip minus 0.2\baselineskip}
Over the past few years, the Internet of Things (IoT) has grown by leaps and bounds, and created a world where smart devices are able to connect to the Internet and communicate with each other. This, hence, fosters an ever-increasing number of real-time monitoring networks in many domains, e.g., traffic monitoring\cite{8571245}\cite{9210202}, underwater acoustic sensor networks\cite{9513315}\cite{9650694}, autonomous navigation\cite{9714880}\cite{9354008}, and unmanned aerial vehicle (UAV) communications\cite{8579209}\cite{9714786}. Different from traditional networks which are mostly designed with a concentration on transmitting messages accurately and efficiently, these IoT based networks are also supposed to react and make decisions based on received data, which means the timeliness of data is vital since timely data strongly supports the network implementation while outdated data causes erroneous decisions or overreactions. In this regard, how to ensure the service facility to receive fresh information is essential for IoT based systems.

In 2011, a new measurement known as age of information (AoI) was introduced to perfectly convey information freshness, which has been defined as the time elapsed since the birth of the newest delivered packet. Thus, AoI performs better for evaluating freshness of the information stream than traditional metrics, such as packet delivery delay and round-trip time (RTT), which only focus on single update. For example, small delay does not mean information is timely since it might suffer from long interval between consecutive deliveries.

However, in some scenarios, the interest is mostly in the information freshness at particular moments instead of time averages.  By noticing this, we proposed a new information freshness measurement termed as age upon decisions (AuD) to characterize information freshness at decision moments when decisions are made according to received update \cite{8720507}. Meanwhile, the IoT based systems that utilize updates to make random decisions are referred to as update-and-decision systems. For example, an IoT based system where the inter-arrival times are deterministic, service time is general, and the decision intervals are exponential, is modeled as an update-and-decision D/G/1-M system.\footnote{We follow and extend the Kendall notation system. In particular, the suffixes '-M' and '-D' denote exponential inter-decision times and deterministic inter-decision times, respectively.} It is clear that AuD characterizes the information freshness at decision moments more effectively than delivery delay and RTT.

\subsection{Motivations}
Our previous work mainly focused on update-and-decision system with infinite queue and studied how the average AuD behaves when considering different update processes, service time distributions or decision intervals. However, the queue length, or equivalently, the buffer size is sometimes limited to reduce waiting time of received updates in many practical IoT networks, hence we are interested in the behaviors of the AuD in bufferless systems.

\begin{exam}
In the self-driving vehicle model, in-vehicle sensors scan the internal and external environment, and generate packets which contain driving information such as speed, coordinate and road conditions. In conventional wireless communication system, these updates are queued to wait for transmission. However, due to the energy constraints, the buffer size of the in-vehicle sensor is normally limited,  which means the sensor nodes have to drop some packets if necessary to avoid the excessive use of node energy. By processing transmitted data, the monitor would first assess situations, and then might or might not make driving decisions. Specifically, the lane detecting sensors collect information about whether the vehicle deviates from the lane and send packets to the vehicle radio module, in which the packets are processed to see if there is the need to adjust steering of the vehicle. Also, the monitor would receive data reported by the distance sensor and the acceleration sensor, and control the vehicle to stop or decelerate to avoid obstacles. Apparently, the packet are generated randomly and contain unpredictable information, the transmission time which affected by variable packet length and channel condition is random and the driving decision-making process that depends on transmitted data is also random. Since the decision process and service process are asynchronous, AuD would be appropriate to characterize the decision timeliness while the AoI can only evaluate the update timeliness.
\end{exam}

\begin{exam}
Limited by the low available system resources and the small-scale energy storage, many low-cost IoT based systems which could make decisions randomly would prefer to fix the decision interval to minimize the computing resources and energy spent in the decision-making-unit (DMU). This can be seen as the trade-off between timeliness and efficiency. That is, to schedule a deterministic decision interval for energy-saving, or to randomize the decision interval for timeliness. One example of such system is the edge computing devices (e.g., smart gateways in Internet of Vehicles) whose computational capacities and energy supply are not as satisfactory as the centralized server, thus the DMU are expected to consume as less computational resources and energy as possible. In this situation, some DMUs would flexibly choose to scrap data regularly to make periodic decisions rather than run continuously as a running process in the background. 
\end{exam}

Therefore, the main focus of this paper is the decision timeliness of update-and-decision system with a length-1 queue. Specifically, we consider the blocking queue, which blocks the incoming updates when busy and only provides service when idle. We shall investigate how the IoT based bufferless system is influenced by different decision intervals. Also, we notice the fact that the system performance can be further optimized at the expense of high-cost network deployment. For example, the stability and quality of communications in autonomous vehicles can be significantly improved by changing from wireless connection to wired connection or increasing transmit power, which is possible to be applied in inter-communication between core components. In the AuD framework, the network deployment corresponds to the schedule of the service time distribution since the service time is defined as the transmission time. Hence we are also interested in whether service scheduling makes any difference to the system performance.

%短距离的无线通信（服务时间几乎是一样的） 或 周期性的信息交换 或 固定长度数据包 或专线？
%\begin{exam}
%In the smart healthcare scenario, wearable device monitors health conditions of a patient and gather health information such as the heart rate, body temperature and pulse. When one of these metrics shows abnormalities, the smart healthcare system would alert caregivers and close family. For example, patients who suffer from high blood pressure would not trigger the alarm if the blood pressure is normal since the system would not react to normal data. However, when emergency such as a fall happens,In this case, the transmission distance between sensors and the monitor is ultra short, hence it is viable to schedule the transmission time to optimize the decision timeliness
%\end{exam}

\subsection{Contributions}
Compared with previous work on the AuD which optimizes the generation and service of updates as well as decision-making process, this work also investigates the effect of buffer existence on the decision timeliness. However, it is not easy to characterize the decision timeliness and utilization of updates for the updating system since they are both jointly affected by buffer existence, generation process and service time of updates as well as decision intervals. To handle this, we study the average AuDs of update-and-decision M/G/1/1-G system by employing a length-1 blocking queue, where arrivals are Poisson distributed and decision interval is exponentially distributed or periodic. Specifically, since the average AuD of M/G/1/1-D bufferless systems cannot be calculated directly, we derive a reasonable approximate closed-form expression whose tightness is examined by simulations. We also propose the metric missing probability to evaluate the probability of updates missed for decisions. Based on these results, we compare and optimize the system performance under different forms of service scheduling and decision interval.
\begin{enumerate}
\item \emph{AuD of Bufferless Update-and-decision Queue:}
We consider the bufferless M/G/1/1-M systems utilizing Poisson decisions and M/G/1/1-D systems where decisions are made periodically. With analyses on the length-1 blocking queue, the general expressions of the average AuD and missing probability are respectively presented in closed-form. 
By applying obtained general expression to the case with a more typical service time, the average AuDs and missing probabilities under three typical service distributions (i.e., uniform/negative exponential/deterministic case) are also obtained, which are meaningful for a wide variety of practical scenarios. It is numerically shown that the considered bufferless queue results in a much lower average AuD than the FCFS infinite queue under heavy system load, and also has better performance in terms of missing probability under all decision rates.

\item \emph{Optimal Service Scheduling:}
For both updating systems with respective exponential or deterministic decision intervals, we theoretically find the optimal service time distribution. In the M/G/1/1-M bufferless system, we show that the deterministic service time performs best in reducing the average AuD and missing probability while a negative exponentially distributed service time performs worst. As for the M/G/1/1-D bufferless system, deterministic service time is also the most suitable for service statistics while the performance ranking of uniformly distributed and negative exponentially distributed service time depends on decision rate.

\item \emph{Optimal Decision Interval:}
We theoretically prove that the average AuD of M/G/1/1-M bufferless system is independent of decision rate while that of corresponding M/G/1/1-D bufferless system decreases for increasing decision rate. Additionally, the missing probability of bufferless system also drops as decision rate climbs. For decision intervals, it is also shown both in theoretical and simulation results that the Poisson decision systems result in a lower average AuD but a higher missing probability than the periodical decision systems with the same decision rate, where we prove that the average AuD of the latter will eventually drop to the same level the former has.
\end{enumerate}

\subsection{Related Work}
The AoI has been exhaustively investigated under different queuing models since its first introduction. For elementary single stream queue, the authors in \cite{6195689} derived the average AoIs of first-come-first-served (FCFS) M/M/1, M/D/1 and D/M/1 systems. Along this line, a more general FCFS G/G/1 queuing model was considered in \cite{8820073}, in which the stationary distributions of the AoI and the peak AoI were studied. Remarkably, the information freshness in FCFS D/G/1 queuing system was characterized by the probability that the AoI or the peak AoI exceeded a certain limit %since it was not viable to find closed-form expressions for them
\cite{8406909} \cite{8691802}. Apart from FCFS discipline, many other service disciplines have been analyzed to keep information as fresh as possible \cite{6310931,7541764,8695040}. For example, a new discipline namely last-come-first-served (LCFS) was proposed in \cite{6310931}, where the average AoIs of LCFS M/M/1 systems with and without preemption were shown to be lower than that in \cite{6195689}. The authors of \cite{7541764} also paid attention to LCFS queue with a focus on Gamma distributed service time. Moreover, the idea of LCFS discipline was further optimized in \cite{8695040} where the last-generated, first-served discipline was introduced to optimize the average AoI.

In addition to service disciplines, different forms of packet management are also effective to improve the information freshness because the queue length of service facility is limited in some networks, especially in IoT based networks \cite{7415972,8945230,8323423,8006504,9478783,9014267}. For instance, the M/M/1/2* system where the newest update would replace the current update in waiting queue achieves better timeliness than the M/M/1/1 and M/M/1/2 blocking systems\cite{7415972}. Following this research, a server waiting scheme was employed to M/G/1/1 and M/G/1/2* systems and successfully improved the average AoI at the expense of the peak AoI \cite{8945230}. Another efficient way to improve the timeliness of M/G/1/1 and M/G/1/2 queuing models is to introduce a packet deadline, which is set to be deterministic or negatively exponential random\cite{8323423}. On the other hand, the authors in \cite{8006504} investigated the average AoIs of M/G/1/1 queuing systems under the consideration of preemption and blocking, while this result was extended to both G/G/1/1 blocking as well as preemptive queuing models in work \cite{9478783}, which also provided an upper bound for the average AoI in a system employing preemptive queue. Particularly, a retransmission scheme was considered in \cite{9014267} to satisfy the need to improve system timeliness under both small generation rate and packet block-length.

More generally, AoI is also suitable for characterizing timeliness of multi-stream systems \cite{8469047,9099557,8406928,9718306}. The average AoIs of M/M/1 multi-stream system under FCFS and LCFS discipline were computed in \cite{8469047}. More than that, the analyses in \cite{9099557} also presented approximate expressions to specify the average AoI of M/M/1 multi-stream system. Combined with packet management, it has been shown that a particular stream in the multi-stream M/G/1/1 preemptive system could be prioritized by increasing its generation rate \cite{8406928}. Meanwhile, a multi-flow M/G/1/1 system with blocking can foster more timely updates than the corresponding M/G/1 system \cite{9718306}. The authors of \cite{8469047} also proposed stochastic hybrid systems (SHS) to calculate the average AoI, which makes it possible to analyze more complex multi-flow queuing models where the expectations of variables are difficult to calculate \cite{9047958,9312180,9252168}. For example, round-robin service scheduling with retransmission has a lower average AoI than stationary randomized service policy \cite{9047958}. In \cite{9312180}, the authors presented the moment generating function (MGF) of the AoI for a general two-stream model and computed the first moment and the second moment of the AoI. Another work focused on the two-stream system showed that the self-preemptive queue achieves the lowest average AoI while the non-preemptive queue results in two streams with the same priority \cite{9252168}.

For different applications and scenarios, the AoI has propagated many variants for measuring performance of updates \cite{9155332,9146773,9137714}. One example is a new metric termed as channel-aware AoI, which characterizes the missing opportunities for delivering an update since the departure of latest successfully delivered update \cite{9155332}. Also, the authors in \cite{9146773} introduced urgency of information (UoI) aiming at specifying the overall freshness of information stream in remote control systems, where the more urgent information has a larger UoI. Meanwhile, the age of incorrect information was proposed and defined as the time elapsed since the last update that brings correct information in \cite{9137714}.
More relevantly, in \cite{8720507}, AuD was first proposed to characterize the information freshness at particular moments, i.e., decision moments. To gain more insight, the average AuDs of G/G/1-M as well as G/M/1-M update-and-decision systems were investigated in \cite{8796659}. Furthermore, the average AuDs were optimized in \cite{8887253} by finding the optimal update generation process as well as decision-making process, where an efficient algorithm was also presented to find the optimal packet generation process. Recently, the authors of \cite{9184001} studied the update-and-decision M/G/1-G system and showed that the average AuD can be minimized by scheduling service time and decision intervals.
Our work in this paper, however, aims to study the update-and-decision system with a bufferless queue and optimize the system performance in two aspects, i.e., the service scheduling and the decision interval.

\subsection{Organization}
The rest of this paper is organized as follows. Section \ref{sec:model} introduces the system model. Section \ref{sec:M} computes the average AuD and missing probability of M/G/1/1-M bufferless system. In Section \ref{sec:D}, we derive the general expression of average AuD in M/G/1/1-D bufferless system and calculate the missing probability. In Section \ref{sec:simula}, numerical results and Monte Carlo simulations are presented as the verification of our theoretical analyses. Section \ref{sec:conclu} concludes our main work.

\section{System Model}\label{sec:model}
Consider an update-and-decision system based on IoT which contains an information generator (e.g., a camera), a service facility (e.g., a wireless channel), and a decision-making-unit (e.g., a central processor). Status updates are generated based on a Poisson generation process with parameter $\lambda$ and subsequently sent to the service facility, in which updates are processed with service rate $\mu$ following a general independent distribution. The system offered load is denoted by $\rho ={\lambda}/{\mu}$. By analyzing the received updates, decisions are made by the DMU at decision rate $\nu$, which are used for controlling or managing. Noticing that the queue length might be limited for some queuing systems, especially for IoT based real-time systems\cite{7166237}\cite{9615359}, we therefore consider a length-1 queue with blocking discipline in this paper, where the waiting queue is not under consideration and only one update can be handled at a time. This means that the updates will only get served if they arrive when the system is idle, otherwise they will be discarded directly. In this way, the system can be abstracted as an update-and-decision M/G/1/1-G bufferless system.%, as shown in Fig. \ref{fig:model1}.

Given that some updates might be discarded directly in the bufferless system, we therefore refer to updates that arrive while the server is idle as successful updates. We use subscript $k$ to indicate the index of the $k^{th}$ successful update. As shown in Fig. \ref{fig:model2}, the $k^{th}$ successful update arrives at time $t_k$ and completes its service at time $t_k^{\\'}$. We also denote: (\romannumeral1) $X_k = t_{k+1} - t_{k}$ as the inter-arrival time between the $k^{th}$ successful update and the next arriving update (which might be dropped directly or delivered successfully), (\romannumeral2) $Y_{k} = t_k^{\\'} - t_{k-1}^{\\'}$ as the departure interval between the $k^{th}$ and ${(k-1)}^{th}$ successful updates, {(\romannumeral3) $T_k = t_k^{\\'}-t_k$ as the period that the $k^{th}$ successful update staying in the system, which is equivalent to its service time $S_k$ in blocking model, (\romannumeral4) $N_k$ as the number of decisions being made based on the ${(k-1)}^{th}$ successful update during $Y_k$, in which $N_k = 0, 1, ...$, and $\tau_{k_j}$ as the decision epoch, in which $j = 1, 2, ..., N_k$, and (\romannumeral5) $Z_j = \tau_j - \tau_{j-1}$ as the decision interval between two neighboring decision epochs.
%\begin{figure*}[!t]
%\centering
%\includegraphics[width=0.45\textwidth]{IoT.png}
%\caption{Update-and-decision system employing a blocking queue.}
%\label{fig:model1}
%\end{figure*}

\begin{figure}[!t]
%\centering
\setlength{\abovecaptionskip}{-0.7cm}
\includegraphics[width=1\linewidth]{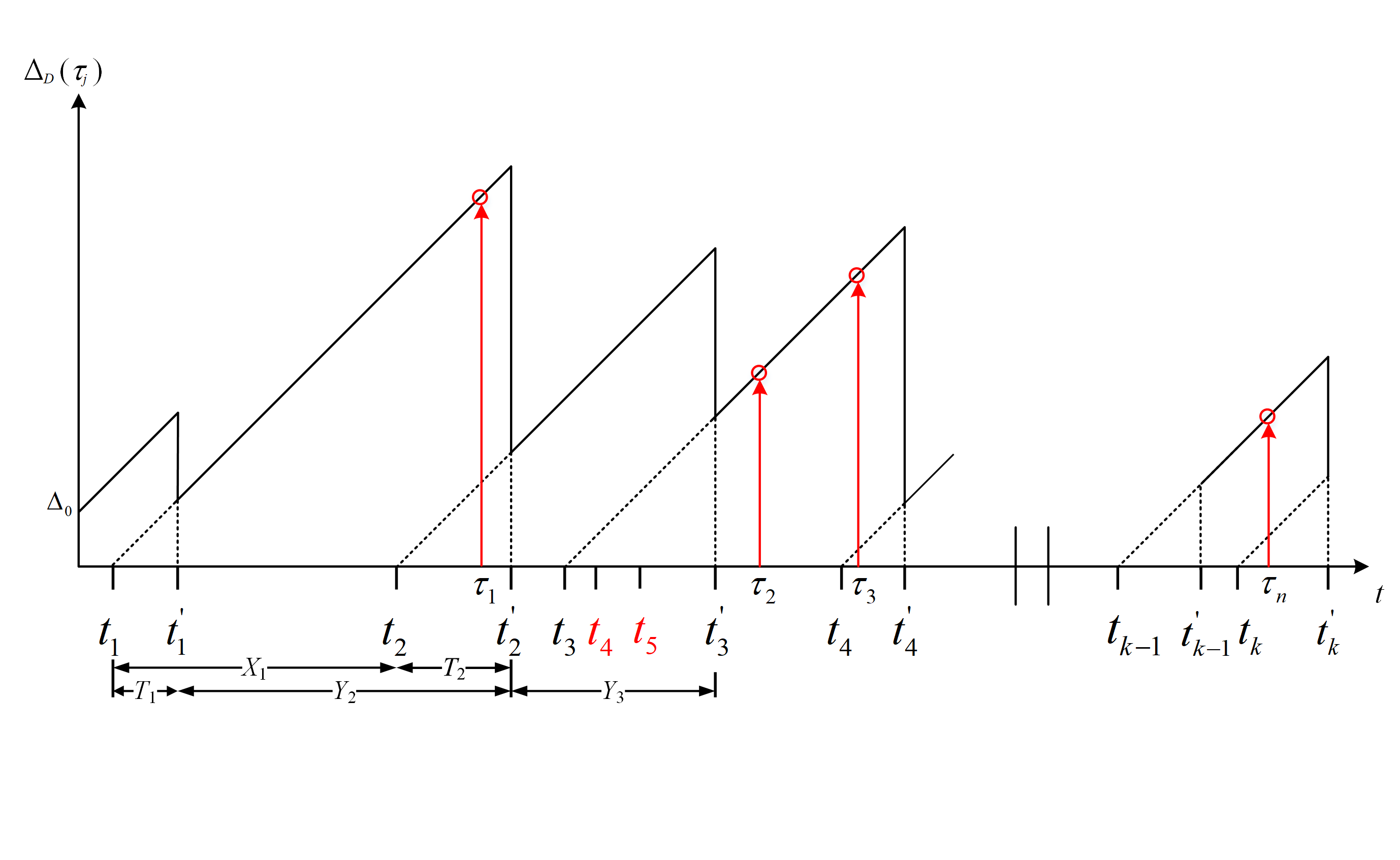}
\caption{Example of age upon decisions for a bufferless system, in which the red abscissas represent the updates dropped in the queue. }
\label{fig:model2}
\end{figure}

%X��S���� %S�ķֲ�

%\label{sec:model}
%\begin{figure*}[!t]
%\centering
%\includegraphics[width=0.85\textwidth]{Protocol.eps}
%\caption{The massive random access network and the frame structure of the AFSD protocol. $X_i$ represents the packet of UE $i$.}
%\label{fig:model}
%\end{figure*}

\begin{defn}(\emph{Age upon Decisions-AuD}\cite{8720507}, \cite{8796659}).
Denote the index of the most recently received update at the $j^{th}$ decision epoch by $N\left ({\tau _{j}}\right)=\text {max}\left \{{k|t_{k}'\leq \tau _{j}}\right \}$,
%\begin{align}
%N_{\rm U}\left ({\tau _{j}}\right)=\text {max}\left \{{k|t_{k}'\leq \tau _{j}}\right \},
%\end{align}
and the generation time of the update by $U\left ({\tau _{j}}\right)=t_{N\left ({\tau _{j}}\right)}$.
%\begin{align}
%U\left ({\tau _{j}}\right)=t_{N_{\rm U}\left ({\tau _{j}}\right)}.
%\end{align}
Then, the \emph{Age upon Decisions} of the updating system is defined as a random process
\begin{align}
\Delta _{\rm D}\left ({\tau _{j}}\right)=\tau _{j}-U\left ({\tau _{j}}\right).
\end{align}
\end{defn}

Apparently, a lower AuD means a more timely decision. By taking the new dimension (i.e., decision process) into consideration, AuD $\Delta _{\rm D}({\tau _{j}})$ assesses the information freshness at epochs when decisions are made while AoI quantifies that at every epoch. It is worth noting that AuD $\Delta _{\rm D}({\tau _{j}})$ would drop to AoI if the decision epoch $\tau _{j}$ is replaced by arbitrary time $t$.

With the above assumptions and definitions, we will focus on the average AuD of the system when given different service time distributions and decision intervals. Suppose $N_T$ decisions are made during the period $\emph T$, with $\mathop {\lim }\limits_{{j} \to \infty } {\tau _{j}} =  + \infty $, the average AuD is expressed as
\begin{align} \overline{\Delta}_{\rm D}=\lim _{T\rightarrow \infty }\frac {1}{ N_{T}} \sum _{j=1}^{N_{{T}}} \Delta _{\rm D}\left ({\tau _{j}}\right).
\end{align}

\section{Average AuD with Blocking Queue and Random Decisions} \label{sec:M}
In this part, we study the update-and-decision M/G/1/1-M bufferless system with a general service time, where both the arrival and decision intervals follow an exponential distribution. At first, we derive the average AuD when given a general service time distribution and three typical service time distributions (i.e., the uniform/negative exponential/deterministic case) separately. After that, we calculate the utilization rate of received updates to evaluate the system efficiency, which is determined by whether successful updates are useful for decision-making. Further, according to the obtained results, we shall theoretically investigate how service time distribution affects the updating system when the DMU makes decisions randomly.

\subsection{Average AuD}
Firstly, for the M/G/1/1-M bufferless system, \emph{Theorem} \ref{thm:1} presents an exact expression for the average AuD.

\begin{thm}\label{thm:1}
For the M/G/1/1-M bufferless system with random decisions, the average AuD is
\begin{align} \label{equ:MG11M}
\overline{\Delta}_{\text {M/G/1/1-M}}^{\rm{blocking}}=\frac {\lambda \mu \mathbb {E}\left [{S^{2}}\right]}{2\left (\lambda+\mu \right )} + \frac{\lambda+\mu}{\lambda \mu}.
\end{align}
\end{thm}
\begin{proof}
See Appendix \ref{app:thm1}.
\end{proof}

For the M/G/1/1-M bufferless system, one can therefore deduce that the average AuD is jointly determined by generation process and service time while is independent of rate of decisions $\nu$ based on \emph{Theorem} \ref{thm:1}.

Now we are ready to calculate the average AuDs of M/G/1/1-M bufferless systems under different service time distributions. In particular, we consider three typical service time distributions, i.e., the uniform distribution, the negative exponential distribution, and the distribution corresponding to deterministic service time. The corresponding update-and-decision systems can be abstracted as the M/U/1/1-M bufferless system, the M/M/1/1-M bufferless system, and the M/D/1/1-M bufferless system. Given that the mean service time $\mathbb E[S] = 1/\mu$, we also denote the corresponding probability density functions (PDFs) as
\begin{align}
f_{\text {\emph{S}U}}\left ({t}\right) &=\frac {\mu }{2},~~~~\quad\quad\quad\quad  t\in \left ({0, \frac{2}{\mu} }\right ), \label{equ:U}\\
f_{\text {\emph{S}E}}\left ({t}\right ) &= \mu e^{-\mu t},~~~\quad\quad\quad  t\geq 0, \label{equ:M}\\
f_{\text {\emph{S}D}}\left ({t}\right) &=\delta \left ({t-\frac {1}{\mu }}\right), ~~\quad t\geq 0 \label{equ:D},
\end{align}
where $\delta(t)$ is the Dirac delta function. For the above bufferless systems, we are able to compute their average AuDs in the following corollary by combining the above assumptions on service time distribution and the result in \emph{Theorem} \ref{thm:1}.

\begin{cor}\label{cor:1}
For the M/U/1/1-M, M/M/1/1-M and M/D/1/1-M bufferless systems with random decisions, the average AuDs are given by
\begin{align}
\overline{\Delta }_{\text {M/U/1/1-M}}^{\rm{blocking}} &= \frac {3+6\rho+5\rho^2}{3\lambda \left ({1+\rho} \right)} \label{equ:MU11M},
\\ \overline{\Delta }_{\text {M/M/1/1-M}}^{\rm{blocking}} &= \frac {1+2\rho+2\rho^2}{\lambda \left ({1+\rho} \right)}\label{equ:MM11M},
\\ \overline{\Delta }_{\text {M/D/1/1-M}}^{\rm{blocking}} &= \frac {2+4\rho+3\rho^2}{2\lambda \left ({1+\rho} \right)}\label{equ:MD11M},
\end{align}
and they have 
\begin{align}
\overline {\Delta }_{\text {M/D/1/1-M}}^{\rm{blocking}}< \overline {\Delta }_{\text {M/U/1/1-M}}^{\rm{blocking}}< \overline {\Delta }_{\text {M/M/1/1-M}}^{\rm{blocking}}
\end{align}
when given the same service rate $\mu$.
\end{cor}

\begin{proof}
From (\ref{equ:U}), (\ref{equ:M}) and (\ref{equ:D}), it is not difficult to calculate the second moment of service time under different distributions, one can get
\begin{align} \label{equ:E}
\mathbb {E}\left [S_{\rm U}^2\right] = \frac{4}{3\mu^2},
\mathbb {E}\left [S_{\rm E}^2\right] = \frac{2}{\mu^2},
\mathbb {E}\left [S_{\rm D}^2\right] = \frac{1}{\mu^2}.
\end{align}
By inserting (\ref{equ:E}) into (\ref{equ:MG11M}), the average AuDs  can be readily obtained. It can also be found that for all the distributions with the same mean, the one that has the smallest $\mathbb  E[S^2]$ will achieve the lowest AuD. 
\end{proof}

%For the above three update-and-decision systems, we are ready to compare the average AuDs between them based on \emph{Corollary} \ref{cor:1}.

%\begin{thm}\label{thm:2}
%For the M/U/1/1-M, M/M/1/1-M and M/D/1/1-M bufferless systems with random decisions when given the same service rate $\mu$, the average AuDs have
%\begin{align}
%\overline {\Delta }_{\text {M/D/1/1-M}}^{\rm{blocking}}< \overline {\Delta }_{\text {M/U/1/1-M}}^{\rm{blocking}}< \overline {\Delta }_{\text {M/M/1/1-M}}^{\rm{blocking}}.
%\end{align}
%\end{thm}
%
%\begin{proof}
%By simplifying (\ref{equ:MU11M}), (\ref{equ:MM11M}) and (\ref{equ:MD11M}), \emph{Theorem} \ref{thm:2} can be readily proved.
%\end{proof}

Hence, it can be concluded that in terms of information freshness at decision moments, M/U/1/1-M bufferless system performs better than M/M/1/1-M bufferless system, but performs worse than M/D/1/1-M bufferless system when given the same service rate $\mu$. For service scheduling, the deterministic service time is therefore the best strategy to achieve the lowest average AuD when given Poisson arrivals and Poisson decisions.

\subsection{Missing Probability of Updates}
In some IoT based scenarios, the focus of the system is not only the information freshness, but also the efficiency of the system. For example, when UAVs perform search and rescue missions in disaster areas, the updates received by the central controller usually consume more transmission resources under harsh conditions (e.g., dusty environment). Whether these updates actually work is extremely important to the efficiency of decision-making process. In this case, how to use the limited number of updates more efficiently becomes the main concern. By observing this, we therefore propose a metric termed as missing probability to characterize the utilization rate of successful updates and thus analyze the efficiency of the M/G/1/1-M bufferless system by computing the missing probability.

In our model, the decision number $N_k$ for the $k^{th}$ successful update is a Poisson distributed random variable during the inter-departure time $Y_k$, which indicates that some updates might complete service but will not be used to make decisions. For example, the second successful update whose generation time is $t_2$ in Fig. \ref{fig:model2} is a missed update because no decision is made during the period $Y_3$. In other words, some updates might be futile. Thus, the more updates being used for making decisions, the more efficient the update-and-decision system will be. The authors of \cite{8887253} and \cite{9184001} proved that although the average AuD of the update-and-decision G/G/1-M system does not reduce as the decision rate $\nu$ increases due to the independence between them, one can still improve the utilization rate of successful updates through increasing decision rate. This is also viable for the G/G/1/1-M bufferless system and will be verified later in numerical results. Specifically, fewer updates will be missed for decisions when increasing the decision rate. The missing probability is specified as follows.

\begin{defn}(\emph{Missing Probability}\cite{8887253})
Missing probability $p_{\rm{mis}}$ of status updates is the limiting ratio between the number of successful updates missed for decisions and the number of total received updates when considering an infinite period.\footnote{\textcolor{black}{Note that the dropped updates are not taken into consideration when defining the missing probability. This is because the dropped updates consume much less system resources than updates received by the receiver, which means updates dropped in the queue are different from successful updates. Since the aim of defining the missing probability is to evaluate the system efficiency, it is more objective to define it as the ratio between updates missed for decisions and total successful updates instead of the ratio between updates missed or dropped and the total generated updates.}}
\end{defn}

\begin{thm}\label{thm:3}
For the M/G/1/1-M bufferless system with random decisions, the missing probability is
\begin{align}\label{equ:pmis}
p_{\rm {mis}}^{\text {M/G/1/1-M}} = G_{X}\left ({-\nu }\right)G_{S}\left ({-\nu }\right),
\end{align}
where $G_{X}(s) = \mathbb {E}[e^{sX}]$ is the MGF of inter-arrival time $X$ and $G_{S}(s) = \mathbb {E}[e^{sS}]$ is that of service time $S$.
\end{thm}

\begin{proof}
It is worth noting that the missing probability $p_{\rm mis}$ is equal to the probability that no decision is made during the departure interval $Y$, i.e., Pr\{$N_k = 0$\}. Given the inter-departure time $Y=y$, $\Pr\{N_k = 0\}=e^{-\nu y}$ follows from the Poisson decision process. Hence, it has
\begin{align}
p_{\rm {mis}}^{\text {M/G/1/1-M}} &= \mathbb {E}\left [{e^{-\nu Y}}\right] \nonumber \\
&= \mathbb {E}\left [{e^{-\nu  X }}\right] \mathbb {E}\left [{e^{-\nu S }}\right]\nonumber \\
&= G_{X}\left ({-\nu }\right)G_{S}\left ({-\nu }\right),
\end{align}
where the first equation follows from taking expectation over $Y$.
\end{proof}

By comparing the missing probability of M/G/1-M updating system with infinite buffer size to that of M/G/1/1-M bufferless system, the following theorem can be readily obtained.
\begin{thm}\label{thm:100}
For the missing probabilities of update-and-decision M/G/1-M and M/G/1/1-M systems, it has
\begin{align}
p_{\rm {mis}}^{\text {M/G/1/1-M}} \leq 	\ p_{\rm {mis}}^{\text {M/G/1-M}}.
\end{align}
\end{thm}
\begin{proof}
According to \cite[Theorem 3]{9184001}, the missing probability of M/G/1-M system is
\begin{align}
p_{\rm {mis}}^{\text {M/G/1-M}}=G_{S}\left ({-\nu }\right) \frac {\rho \nu +\lambda }{\lambda +\nu }.
\end{align}

Considering the fact that the M/G/1/1-M bufferless system has $G_{X}({-\nu }) = \lambda/(\lambda + \nu)$, it is clear that the missing probability of the M/G/1-M system is larger than and will reduce to the same level as the M/G/1/1-M bufferless system has if system load $\rho \to 0$.
\end{proof}

Likewise, for the bufferless M/G/1/1-M systems with typical service time distributions, we shall calculate their missing probabilities to characterize the system efficiency. With the same assumptions on the service time distribution in previous subsection, we have \emph{Corollary} \ref{cor:2}.

\begin{cor}\label{cor:2}
For the M/U/1/1-M, M/M/1/1-M and M/D/1/1-M bufferless systems with random decisions, the missing probabilities are given by
\begin{align}
p_{\rm {mis}}^{\text {M/U/1/1-M}} &= \frac{\lambda \mu}{2\nu \left( \lambda + \nu \right)}\left( 1-e^{-2 \frac{\nu}{\mu}} \right),
\\p_{\rm {mis}}^{\text {M/M/1/1-M}} &= \frac{\lambda \mu}{\left( \lambda + \nu \right)\left( \mu + \nu \right)},
\\ p_{\rm {mis}}^{\text {M/D/1/1-M}} &= \frac{\lambda}{\lambda + \nu}e^{-\frac{\nu}{\mu}}.
\end{align}
\end{cor}

\begin{proof}
From (\ref{equ:U}), (\ref{equ:M}) and (\ref{equ:D}), one can get
\begin{align}  \label{equ:G}
G_{\rm \emph{S}U}\left ({-\nu }\right) =& \frac{\mu }{2\nu}\left( 1-e^{-2 \frac{\nu}{\mu}} \right), \nonumber \\
G_{\rm \emph{S}E}\left ({-\nu }\right) =& \frac{\mu}{\mu + \nu},\nonumber \\
G_{\rm \emph{S}D}\left ({-\nu }\right) =& e^{-\frac{\nu}{\mu}}.
\end{align}
Also note that it has $G_{X}({-\nu }) = \lambda/(\lambda + \nu)$ in the M/G/1/1-M bufferless system. By combining (\ref{equ:pmis}) and (\ref{equ:G}), \emph{Corollary} \ref{cor:2} can be readily obtained.
\end{proof}

The relationship between missing probabilities in \emph{Corollary} \ref{cor:2} is given in the following theorem.

\begin{thm} \label{thm:101}
For the M/U/1/1-M, M/M/1/1-M and M/D/1/1-M bufferless systems with random decisions when given the same service rate $\mu$ as well as rate of decisions $\nu$, it has
\begin{align}
p_{\rm {mis}}^{\text {M/D/1/1-M}} < p_{\rm {mis}}^{\text {M/U/1/1-M}} < p_{\rm {mis}}^{\text {M/M/1/1-M}}.
\end{align}
\end{thm}
\begin{proof}
See Appendix \ref{app:thm101}.
\end{proof}

Based on \emph{Theorem} \ref{thm:101}, it is clear that for the given arrival rate $\lambda$, service rate $\mu$ and decision rate $\nu$, the deterministic service time achieves the lowest missing probability while the negative exponentially distributed service time achieves the highest one. In conclusion, deterministic service time reduces more wastage of the system resources and allows the system to be more efficient in the M/G/1/1-M bufferless system. Also, by taking \emph{Corollary} \ref{cor:1} into account, it can be found that a deterministic service time is more suitable than other random service time distributions for the update-and-decision system with a length-1 blocking queue and random decisions, which will be further verified later in Section \ref{sec:simula}.

\section{Average AuD with Blocking Queue and Periodic Decisions}\label{sec:D}
Although it has been shown that the deterministic service time ensures better system performance for Poisson decision system, it is still not clear which decision interval is better: random or deterministic? Therefore, in this part, we study the bufferless system performance when the decision interval is deterministic. We first derive the average AuD of the M/G/1/1-D bufferless system, based on which, we then calculate and compare the average AuDs of M/G/1/1-D bufferless systems with different types of service time distribution, i.e., the uniform/negative exponential/deterministic distributions. Finally, similar to the case of M/G/1/1-M updating system, we characterize the system efficiency by computing the missing probability $p_{\rm mis}$.

Since the decision process is no longer a Poisson process and the decision events are no longer uniformly distributed on the interval, the average AuD cannot be calculated directly. To handle this, we make an approximation that the decision epochs still follow a uniform distribution within every departure interval $Y_k$ in this case, which only has a negligible impact on the accuracy of the results and its accuracy will be verified later in simulations. Moreover, we assume that $\nu$ is an integer multiple of $\mu$ to ensure the decision supplies, i.e., $\nu = m_0\mu$. Also, aiming at controlling the missing probability to remain in a low level and be observable, $m_0$ is assumed to be larger than unity, i.e., $m_0 \geq 1$.

\subsection{Average AuD and Missing Probability}
Since the inter-arrival time $X_k$ follows an exponential distribution and inter-decision time $Z_j$ is fixed, we have their PDFs $f_{X}(t) = \lambda e^{-\lambda t}$ and $f_{Z}(t) = \delta \left ({t-1/\nu }\right)$.

\begin{thm} \label{thm:5}
For the M/G/1/1-D bufferless system with deterministic decisions, the average AuD is approximated by
\begin{align}\label{equ:MG11D}
&\overline{\Delta}_{\text {M/G/1/1-D}}^{\rm{blocking}} \nonumber \\
&=  \frac {\mathbb {E}\left [{T_{k-1}}\right]\left (\mathbb {E}\left [{N_{k}^{(1)}}\right]+\mathbb {E}\left [{N_{k}^{(2)}}\right] \right)}{\nu\mathbb {E}\left [{Y_k}\right]}\nonumber \\
& \; \quad +\frac {\mathbb {E}\left [{\left ({N_{k}^{(1)}}\right)^{2}}\right]+\mathbb {E} \left [{\left ({N_{k}^{(2)}}\right)^{2}}\right]+2{\mathbb {E}\left [{N_{k}^{(1)}}\right] \mathbb {E}\left [{N_{k}^{(2)}}\right]}}{2{\nu^2}\mathbb {E}\left [{Y_{k}}\right]},
\end{align}
where $N_k^{(1)}$ and $N_k^{(2)}$ are the numbers of decisions made before and after the arrival epoch $t_k$ within the corresponding inter-departure time $Y_k$.
\end{thm}
\begin{proof}
See Appendix \ref{app:thm5}.
\end{proof}

Similarly, by replacing the general service time by typical service times we have mentioned in previous section, we investigate the average AuDs of the M/U/1/1-D, M/M/1/1-D as well as M/D/1/1-D bufferless systems, and the impact of service time distribution on system timeliness.

\begin{cor}\label{cor:3}
For the M/U/1/1-D, M/M/1/1-D and M/D/1/1-D bufferless systems with deterministic decision intervals, the average AuDs are approximated by
\begin{align} \overline{\Delta}_{\text {M/U/1/1-D}}^{\rm{blocking}}=&\frac {2\rho m_0 + 4m_0 +\rho + 1}{2\mu m_0 \left ({1+\rho }\right) } +\frac {\left ({1+\beta}\right)}{2\mu m_{0}\left ({1+\rho }\right)\left ({1-\beta}\right)} \nonumber\\
&+\frac {\rho \left ({8m_0^3+18m_0^2+13m_0+3}\right)}{12\mu m_{0}^3\left ({1+\rho}\right)}, \label{equ:MU11D}\\
\overline{\Delta}_{\text {M/M/1/1-D}}^{\rm{blocking}}=&\frac {2+\rho}{\mu \left ({1+\rho }\right)}+ \frac {\rho \left ({1+\alpha}\right)}{2\mu m_{0}\left ({1+\rho }\right)\left ({1-\alpha}\right)} \nonumber \\ 
&+\frac {\left ({1+\beta}\right)}{2\mu m_{0}\left ({1+\rho }\right)\left ({1-\beta}\right)}, \label{equ:MM11D}\\
\overline{\Delta}_{\text {M/D/1/1-D}}^{\rm{blocking}}=&\frac {4+3\rho}{2\mu \left ({1+\rho }\right)}+\frac {\left ({1+\beta}\right)}{2\mu m_{0}\left ({1+\rho }\right)\left ({1-\beta}\right)},\label{equ:MD11D}
\end{align}
in which $\alpha = e^{-\mu / \nu}$ and $\beta = e^{-\lambda / \nu}$.
\end{cor}

\begin{proof}
See Appendix \ref{app:cor3}.
\end{proof}

According to \emph{Corollary} \ref{cor:3}, we prove \emph{Theorem} \ref{thm:6} for service scheduling, which shows the average AuD relationship between the above three systems.

\begin{thm}\label{thm:6}
For the M/U/1/1-D, M/M/1/1-D and M/D/1/1-D bufferless systems with deterministic decision intervals when given the same service rate $\mu$ as well as decision rate $\nu$, there exists a positive integer $m_{0}^{*}$ such that
\begin{align}
\begin{cases}
\overline{\Delta}_{\text {M/D/1/1-D}}^{\rm{blocking}} < \overline{\Delta}_{\text {M/U/1/1-D}}^{\rm{blocking}} < \overline{\Delta}_{\text {M/M/1/1-D}}^{\rm{blocking}} &\text {if}~m_{0}> m_{0}^{*},
\\ \overline{\Delta}_{\text {M/D/1/1-D}}^{\rm{blocking}} < \overline{\Delta}_{\text {M/M/1/1-D}}^{\rm{blocking}} < \overline{\Delta}_{\text {M/U/1/1-D}}^{\rm{blocking}} &\text {if}~m_{0}\leq m_{0}^{*},
\end{cases}
\end{align}
\end{thm}

\begin{proof}
See Appendix \ref{app:thm6}.
\end{proof}

For the system with a length-1 blocking queue and deterministic decision intervals, \emph{Theorem} \ref{thm:6} shows that for any decision rate, the system with a deterministic service time makes most timely decisions.

Remarkably, by comparing the average AuDs in \emph{Corollary} \ref{cor:3} to \emph{Corollary} \ref{cor:1}, \emph{Theorem} \ref{thm:7} is also obtained for decision scheduling.
\begin{thm}\label{thm:7}
For the M/G/1/1-G bufferless systems when given the same service rate $\mu$ as well as decision rate $\nu$, it has
\begin{align}
\overline{\Delta}_{\text {M/U/1/1-D}}^{\rm{blocking}}\geq&\overline{\Delta}_{\text {M/U/1/1-M}}^{\rm{blocking}}, \\
\overline{\Delta}_{\text {M/M/1/1-D}}^{\rm{blocking}}\geq&\overline{\Delta}_{\text {M/M/1/1-M}}^{\rm{blocking}}, \\
\overline{\Delta}_{\text {M/D/1/1-D}}^{\rm{blocking}}\geq&\overline{\Delta}_{\text {M/D/1/1-M}}^{\rm{blocking}},
\end{align}
in which the equal sign holds as $m_0 \to \infty$.
\end{thm}

\begin{proof}
See Appendix \ref{app:thm7}.
\end{proof}

From \emph{Theorem} \ref{thm:7}, one can conclude that for the M/G/1/1-G updating system with a bufferless queue, the memoryless random decision interval achieves better timeliness than the deterministic decision interval when given the same service time distribution. Therefore, we prefer random decisions to periodic decisions when scheduling decision intervals if the main concern is timeliness of received updates.

Next, in order to evaluate the system efficiency or, equivalently, the utilization of successful updates, we have the following corollary about the missing probability in the M/G/1/1-D bufferless system.
\begin{cor}\label{cor:4}
For the M/U/1/1-D, M/M/1/1-D and M/D/1/1-D bufferless systems with deterministic decision intervals, missing probabilities can be given by
\begin{align}
p_{\rm {mis}}^{\text {M/U/1/1-D}}&=\frac{1}{4m_0}+\frac {{m_{0} \left ({1-\beta}\right)-\rho }}{2\rho ^{2}}, \\
p_{\rm {mis}}^{\text {M/M/1/1-D}}&=\frac{m_0 \left ( \alpha- \beta \right )}{\rho \left( \rho-1 \right)}
+\frac{\left ( \rho-m_0- \rho m_0 \right ) \left (1-\alpha \right )}{\rho}+{\alpha},
\\ p_{\rm {mis}}^{\text {M/D/1/1-D}} &= 0.
\end{align}
\end{cor}

\begin{proof}
See Appendix \ref{app:cor4}.
\end{proof}

Therefore, it is clear that the deterministic service time achieves better performance than other common service time distributions with regard to missing probability, based on which and \emph{Theorem} \ref{thm:6}, one can conclude that deterministic service time is also the most suitable service statistics for update-and-decision M/G/1/1-D bufferless system.

\section{Simulation Results} \label{sec:simula}
\begin{figure}[!t]
%\centering
\includegraphics[width=1\linewidth]{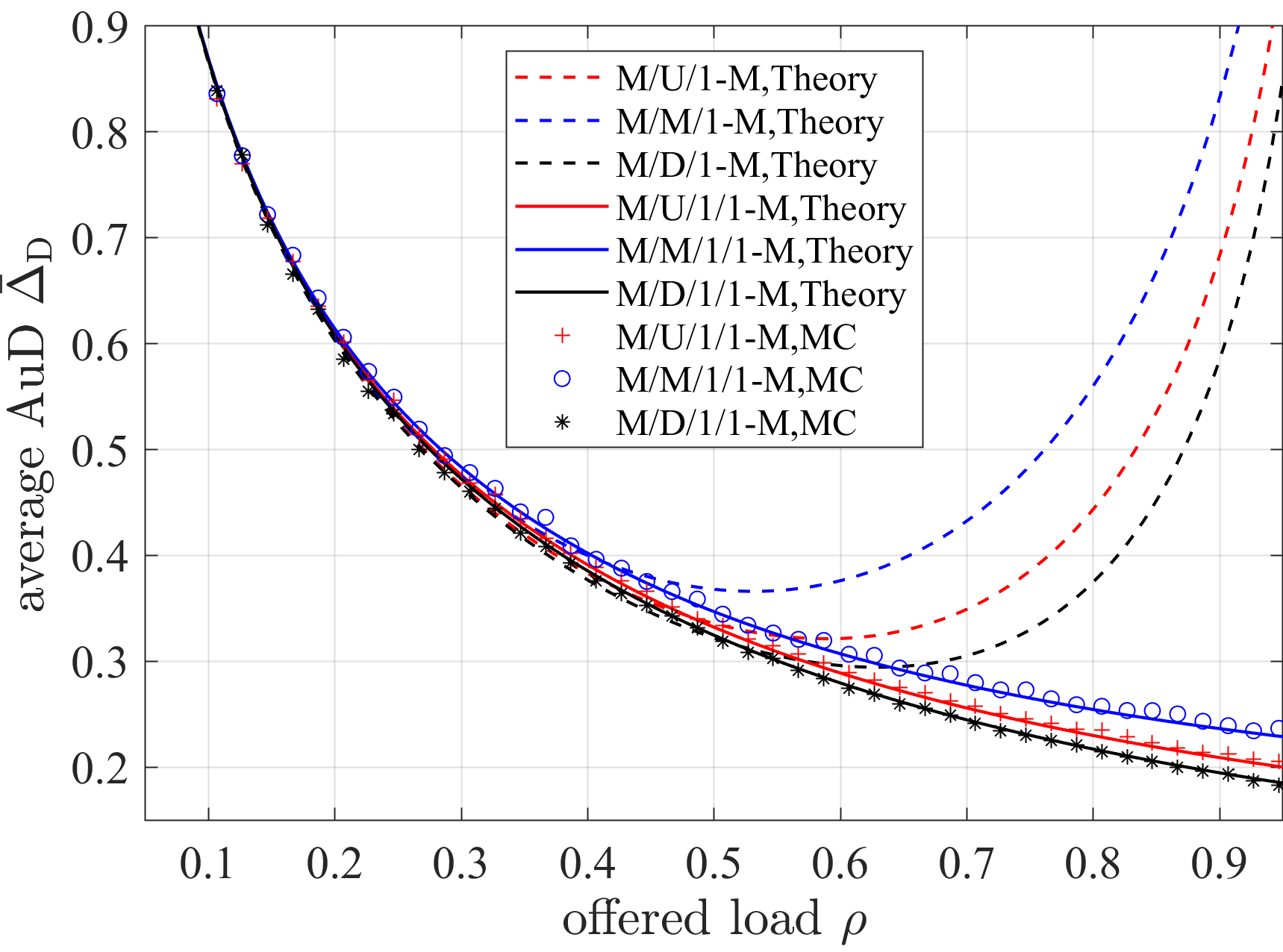}
\caption{Comparison between the average AuDs of update-and-decision M/G/1-M and M/G/1/1-M systems when setting $\mu = 1.5$.}
\label{fig:1}
\end{figure}

In this part, we first demonstrate the performance of the M/G/1/1-M and M/G/1/1-D bufferless systems by presenting numerical results, which verifies our analyses on service scheduling. Then, we compare them to the corresponding systems with infinite buffer size to gain insight on the effect of buffer existence on system performance. Also, we shall find the optimal decision interval for length-1 blocking queue by making an average AuDs comparison between Poisson decisions and periodic decisions. Finally, we make a comprehensive comparison of system performance for above systems with and without buffer. Monte Carlo simulation results, which are marked with the abbreviation MC in the figures, are also provided as a validation check and exactly match with theoretical results.

\begin{figure}[!t]
\centering
\includegraphics[width=1\linewidth]{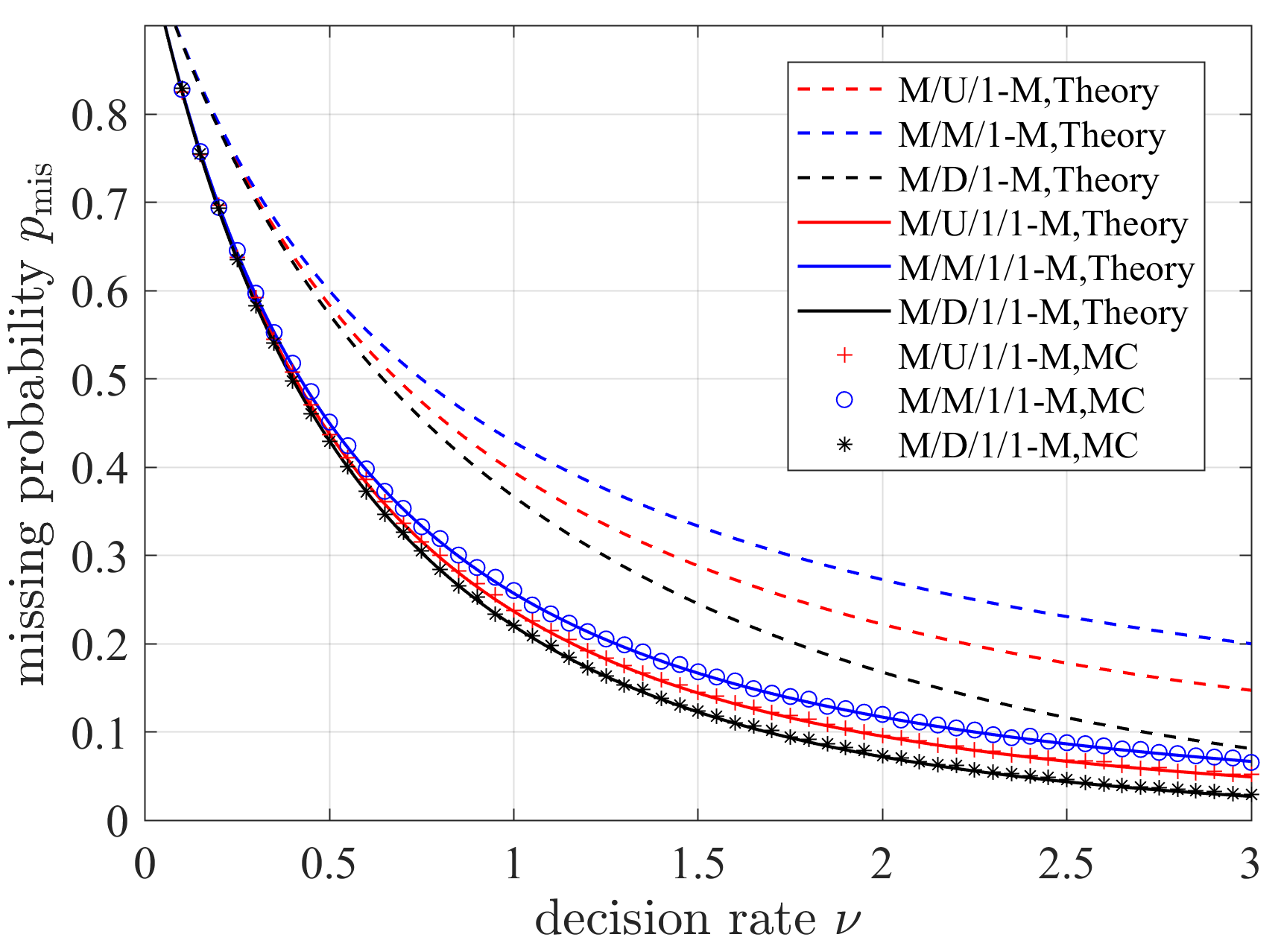}
\caption{Comparison between the missing probabilities of update-and-decision M/G/1-M and M/G/1/1-M systems when setting $\rho = 0.5$.}
\label{fig:2}
\end{figure}

Fig. \ref{fig:1} plots the average AuDs of M/U/1/1-M, M/M/1/1-M and M/D/1/1-M bufferless systems under different system loads $\rho$, where service rate is fixed to 1.5. First, it is observed that the average AuDs drop with the offered load $\rho$ and are minimized when the offered load is relatively high. Second, it can be seen that the deterministic service time achieves the lowest AuD, which is in agreement with the results in \emph{Corollary} \ref{cor:1}. Also, compared with systems employing infinite buffer size whose AuDs are indicated by dotted lines, the bufferless systems improve the average AuD when $\rho > 0.5$ and significantly reduce it under heavy offered load. This is because for the FCFS infinite queue, heavy system load can lead to a longer waiting time, which, however, can be ignored in length-1 blocking queue.

\begin{figure}[!t]
\centering
\includegraphics[width=0.98\linewidth]{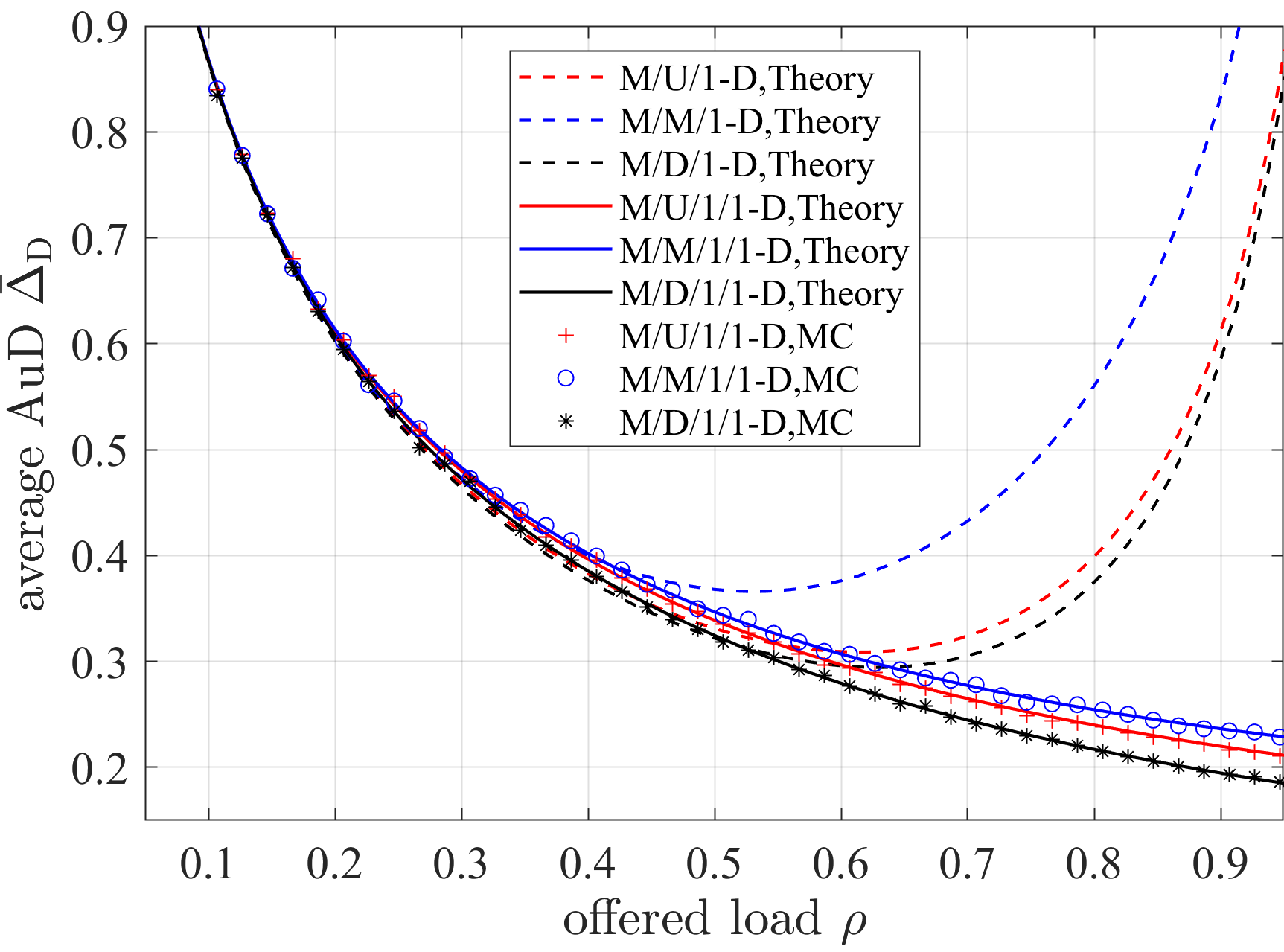}
\caption{Average AuD comparison between M/G/1-D and M/G/1/1-D update-and-decision systems when settings $\mu = 1.5$ and $m_0 = 30$.}
\label{fig:3}
\end{figure}

\begin{figure}[!t]
\centering
\includegraphics[width=1\linewidth]{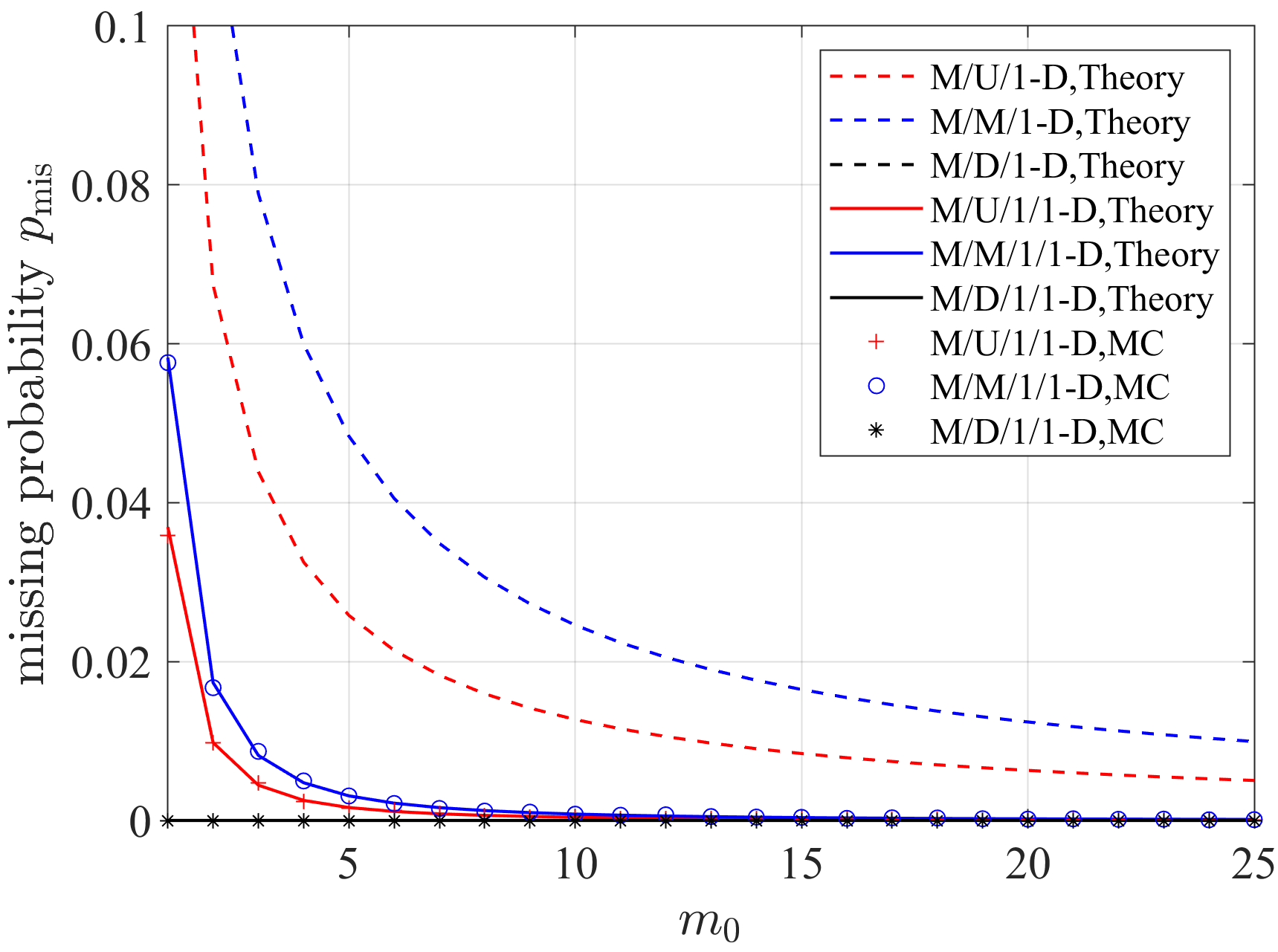}
\caption{Missing probability comparison between M/G/1-D and M/G/1/1-D update-and-decision systems when setting $\rho = 0.5$. }
\label{fig:4}
\end{figure}

Fig. \ref{fig:2} depicts the relation between the missing probability and the increasing decision rate $\nu$ for common offered load $\rho = 0.5$. First, one can see that the missing probability of bufferless system drops as $\nu$ increases, which means that a higher decision rate can contribute to a more effective system. In addition, it can be observed that the system under deterministic service times achieves the lowest missing probability, which indicates that scheduling a deterministic service time is also the best strategy to use when it comes to missing probability. Further, it is also illustrated that the bufferless system has a lower missing probability than corresponding system with infinite buffer size, which coincides with our analyses in \emph{Theorem} \ref{thm:100}.
By combining the results in Fig. \ref{fig:1} and Fig. \ref{fig:2}, it can be deduced that the M/G/1/1-M bufferless  system performs best when the service time is deterministic rather than random.

Fig. \ref{fig:3} illustrates the average AuDs of bufferless systems with periodic decisions with respect to offered load $\rho$. It is observed that the Monte Carlo simulations match our theoretical results well, which shows that our approximation in the calculation of the average AuD is viable and only leads to a trivial deviation. Also, similar to the analysis on Fig. \ref{fig:1}, one can check that the average AuDs decrease with $\rho$ and are far smaller than those of systems with infinite buffer size. We also notice that the average AuD of system scheduling a deterministic service time is lower than that of system applying other service time distributions.

\begin{figure}
\centering
\subfigure[Average AuDs of bufferless systems under different assumptions on service time and decision interval.]{
\includegraphics[width=1\linewidth]{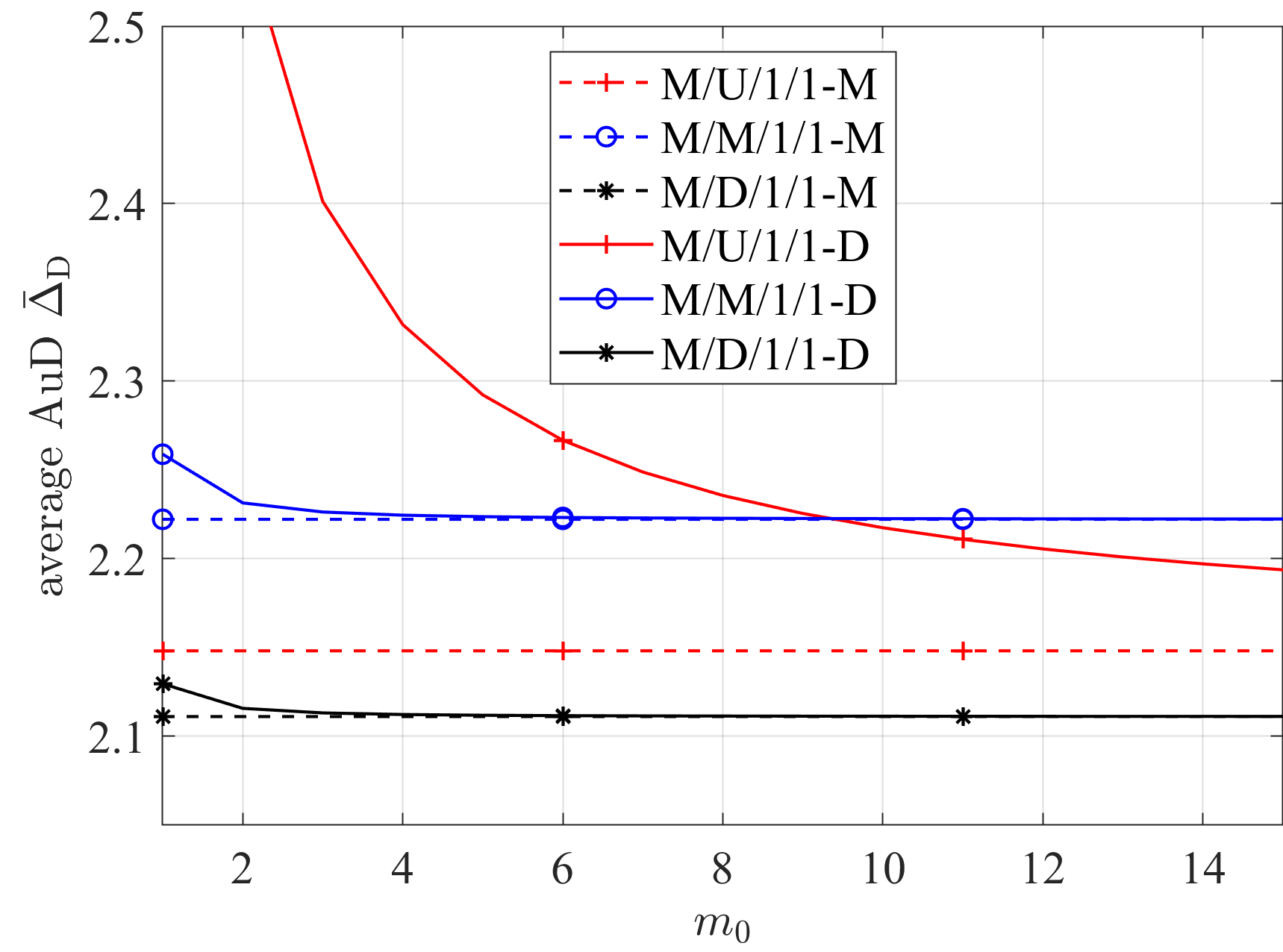}
\label{fig:5}
}
\subfigure[Missing probabilities of bufferless systems under different assumptions on service time and decision interval.]{
\includegraphics[width=1\linewidth]{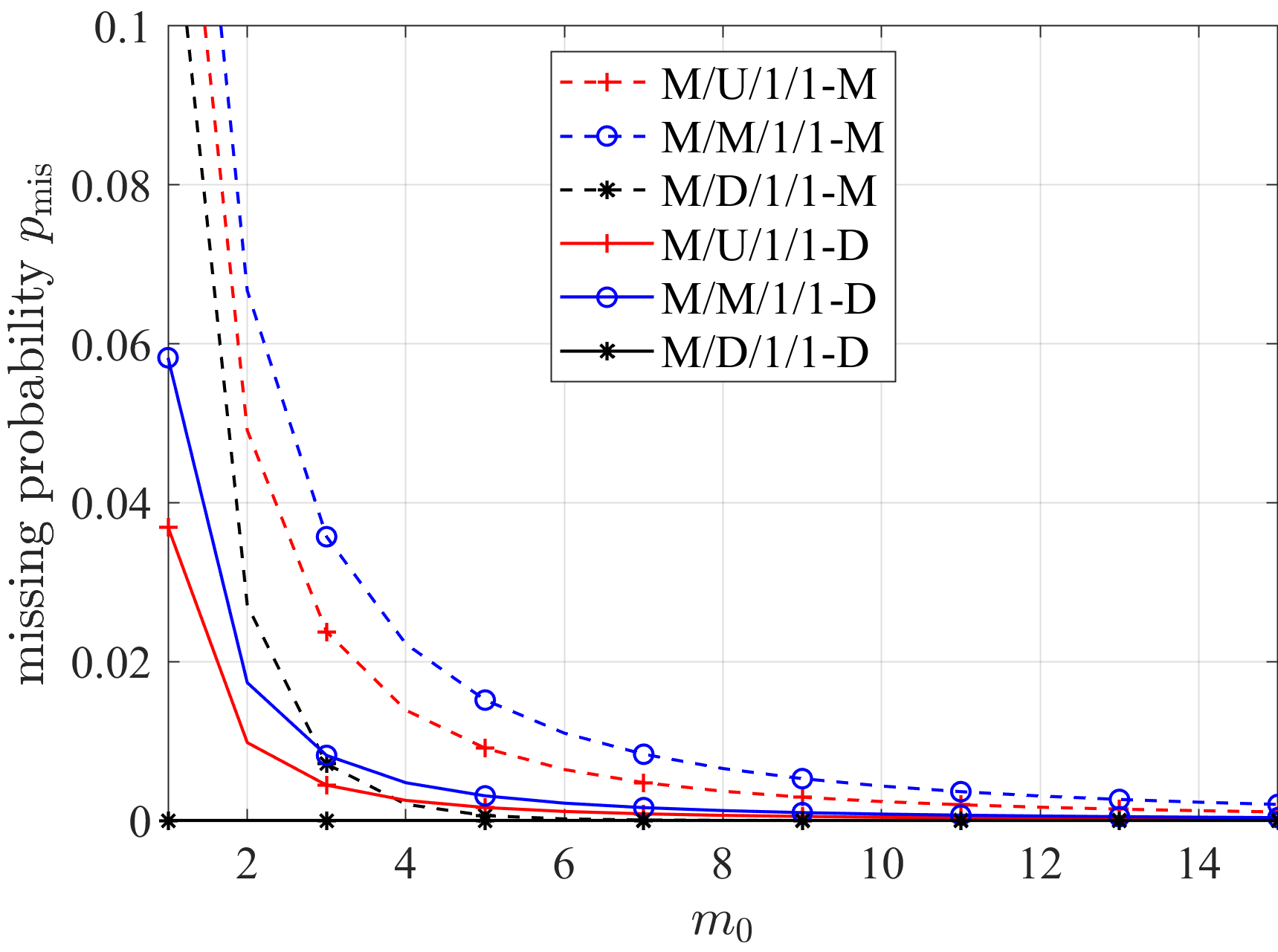}
\label{fig:6}
}
\DeclareGraphicsExtensions.
\caption{Performance comparison between systems with different service times and decision intervals.}
%\label{V9-1-vary-S}
\label{fig:decisions}
\end{figure}

Fig. \ref{fig:4} depicts the missing probabilities of bufferless systems with periodic decisions with respect to decision rate $\nu$. Note that the dashed black line is overlapped by the solid black line, which shows the missing probabilities of M/D/1-D and M/D/1/1-D systems are both equal to 0 and every received update is used for at least one decision. One can also see that the bufferless system achieves a lower missing probability than system with infinite buffer size for all decision rates. This can be explained by considering that the latter has a smaller expectation of inter-departure time $\mathbb E[Y]$, which increases the probability that decision interval is larger than inter-departure time. In addition, one can also observe that the missing probability of system applying a deterministic service time is smaller than that of system applying other service time distributions.

\begin{figure}
\centering
\subfigure[3D stacked bar chart for the average AuD. Note that when the service time is exponentially distributed or periodic, the AuD with exponential decision intervals is slightly larger than the AuD with periodic decision intervals.]
{
\includegraphics[width=1\linewidth]{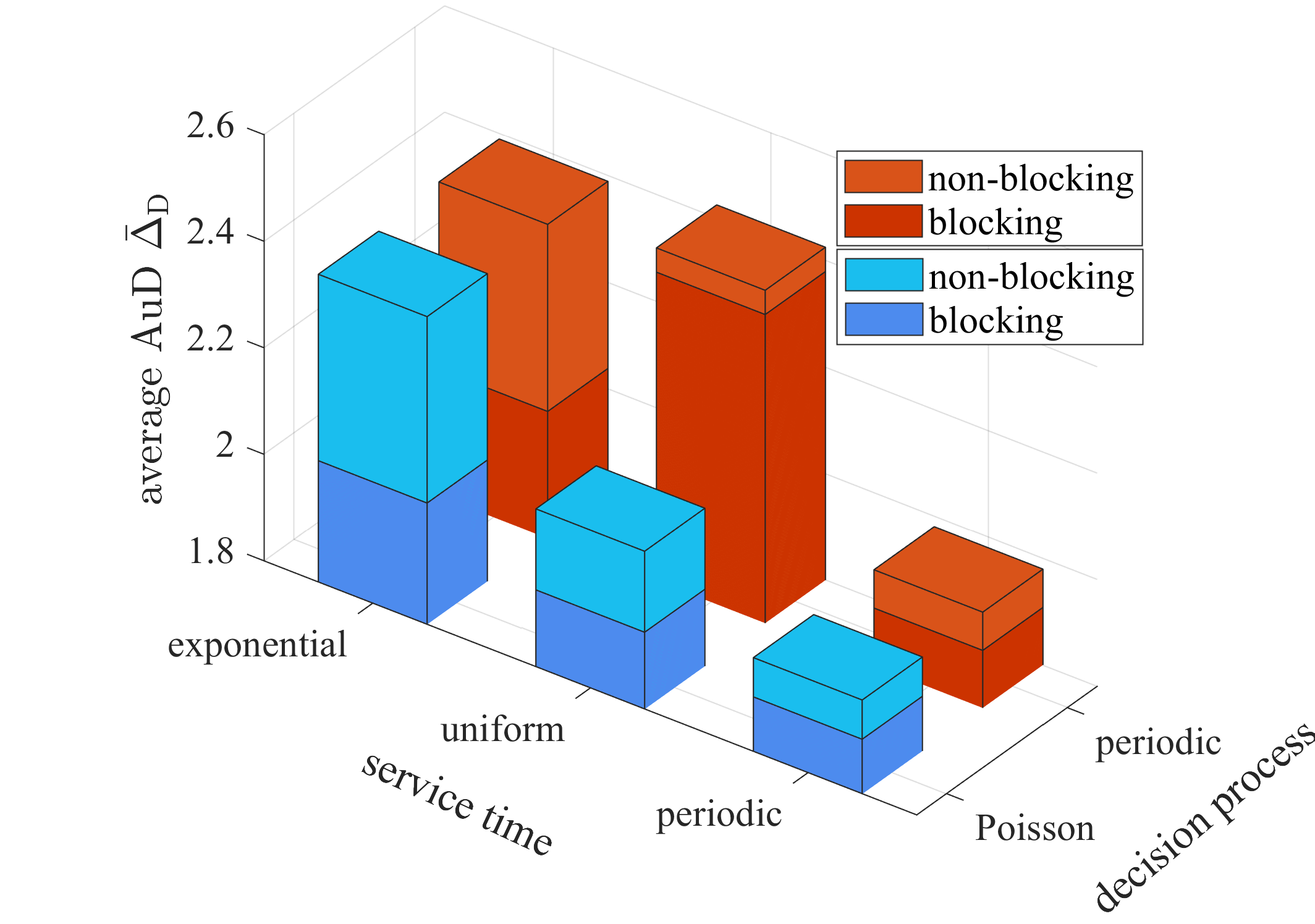}
\label{fig:bar_AuD}
}
\subfigure[3D stacked bar chart for the missing probability.]{
\includegraphics[width=1\linewidth]{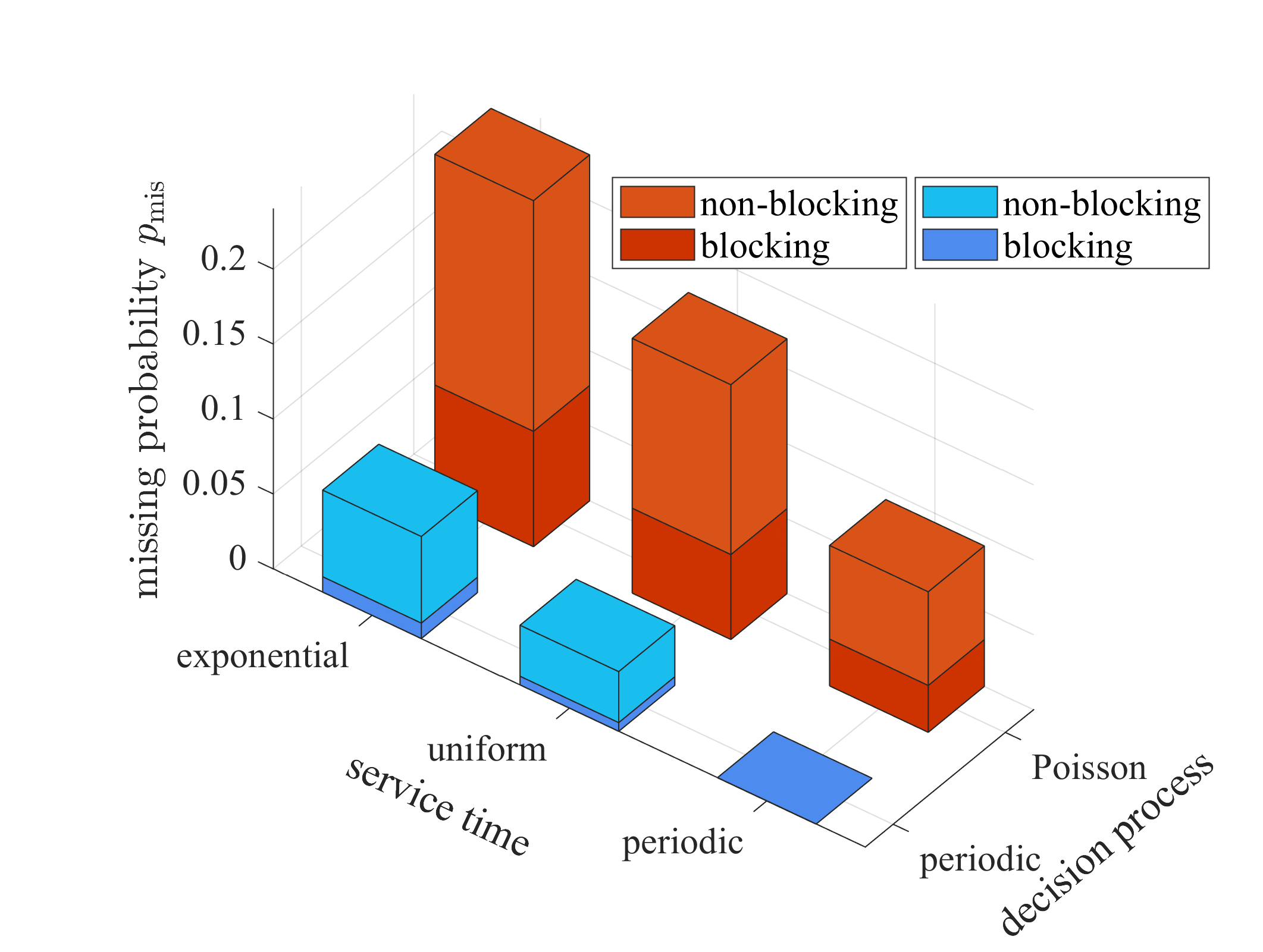}
\label{fig:bar_mis}
}
\DeclareGraphicsExtensions.
\caption{System performance for update-and-decision systems, where $\rho = 0.6$.}
%\label{V9-1-vary-S}
\label{fig:bar}
\end{figure}

By noticing the similarity between Fig. \ref{fig:1}, Fig. \ref{fig:2} and Fig. \ref{fig:3}, Fig. \ref{fig:4}, it can be found that, in terms of buffer existence, the length-1 blocking queue improves both timeliness and efficiency of the updating systems. Further, among the three common service time distributions, the deterministic service time performs best since it can reduce the average AuD as well as the missing probability at the same time.

In Fig. \ref{fig:decisions}, we analyze the performance difference between update-and-decision bufferless systems with Poisson decisions and periodic decisions, where the offered load $\rho = 0.5$. From Fig. \ref{fig:5}, one can see that the average AuDs of systems utilizing periodic decisions decrease with decision rate $\nu$ and drop to those of systems with Poisson decisions as $\nu$ goes to infinity, which verifies \emph{Theorem} \ref{thm:5}. Moreover, one can observe that Poisson decision process remarkably reduces the average AuD, especially when decision rate is small. Fig. \ref{fig:6} shows how the missing probability changes with rate of decisions, where one can find that periodic decisions significantly reduce the missing probability compared with Poisson decisions. Thus, exponential decision interval is more suitable for improving timeliness while periodic decision interval is recommended to boost the efficiency.

Finally, by combining the results on systems with infinite buffer size in\cite{9184001}, we comprehensively compared the average AuDs and missing probabilities of updating systems under different assumptions of buffer existence, service distribution and decision interval in Fig. \ref{fig:bar}. It is intuitively clear that, to satisfy the need of better system performance, the length-1 blocking discipline and deterministic service time are preferred while whether to utilize deterministic or exponential decision intervals depends on the preference between improving timeliness and boosting efficiency.

%��ȥʽ
\section{Conclusion} \label{sec:conclu}
We analyzed both the timeliness and efficiency for the update-and-decision bufferless system based on IoT with Poisson arrivals, in which we focused on how the buffer existence, service time distribution, and decision interval influence system performance. In particular, we considered a length-1 blocking queue where there is no waiting space for incoming updates. For systems with Poisson decisions or periodic decisions, we respectively calculated the average AuD and missing probability under general service time distributions. It is shown that the average AuD would drop with the offered load and achieve its minimum if the offered load is relatively high, and missing probability would be large when the decision rate is relatively small and decline with decision rate.
From the perspective of buffer existence, the bufferless queue effectively improves the system timeliness and efficiency.
Also, in terms of service scheduling, it is found that deterministic service time results in better performance compared with other service time distributions.
As for decision interval, it is found that the Poisson decision system has a lower average AuD but a larger missing probability than periodical decision system.

%future work

%��¼
\appendix
\subsection{Proof of Theorem \ref{thm:1}}\label{app:thm1}
According to \cite{8887253}, the average AuD of the update-and-decision G/G/1/1-M system is
\begin{align} \label{equ:GG11M1}
\overline{\Delta}_{\text {M/G/1/1-M}}^{\rm{blocking}}=\frac {\mathbb {E}\left [{Y_{k}^{2}}\right]+ 2\mathbb {E}\left [{T_{k-1}Y_{k}}\right]}{2\mathbb {E}\left [{Y_{k}}\right]}.
\end{align}
Since our interest is the steady-state behavior of system, the time indices are dropped in this case. Therefore, (\ref{equ:GG11M1}) becomes at steady state and can be rewritten as
\begin{align} \label{equ:GG11M2}
\overline{\Delta}_{\text {M/G/1/1-M}}^{\rm{blocking}}=\frac {\mathbb {E}\left [{Y^{2}}\right]+ 2\mathbb {E}\left [{T}\right]\mathbb {E}\left [{Y}\right]}{2\mathbb {E}\left [{Y}\right]},
\end{align}
where the system time of the ${(k-1)}^{th}$ successful update $T_{k-1}$ and the inter-departure time of the $k^{th}$ successful update $Y_{k}$ are independent with each other.

For the M/G/1/1 blocking system, \cite{8006504} shows that the mean inter-departure time $Y$ is
\begin{align} \label{equ:A1}
\mathbb {E}\left [{Y}\right] = \mathbb {E}\left [{X}\right] + \mathbb {E}\left [{S}\right] ,
\end{align}
and
\begin{align}  \label{equ:A2}
Y = Z + S,
\end{align}
where $Z$ is the residual inter-arrival time until a new update is generated and is exponentially distributed.

Recall that the mean system time is equal to the mean service time, i.e., $\mathbb {E}[T] = \mathbb {E}[S]$. Since the mean inter-arrival time $\mathbb {E}\left [{X}\right] = 1/\lambda$ and mean service time $\mathbb {E}\left [{S}\right] = 1/\mu$, by substituting (\ref{equ:A1}) and (\ref{equ:A2}) into (\ref{equ:GG11M2}), the average AuD of M/G/1/1-M system employing a length-1 blocking queue with random decisions is obtained as
\begin{align}
\overline{\Delta}_{D}=&\frac {\mathbb {E}\left [{Y^{2}}\right]+ 2\mathbb {E}\left [{T}\right]\mathbb {E}\left [{Y}\right]}{2\mathbb {E}\left [{Y}\right]}\nonumber \\
=& \frac {\mathbb {E}\left [{\left({X+S} \right)^{2}}\right]}{2\left ( \mathbb {E}\left [{X}\right] + \mathbb {E}\left [{S}\right] \right )} + \mathbb {E}\left [{S}\right]\nonumber \\
\label{equ:A3}=& \frac {2\mathbb {E}\left [{X}\right]^2 + 2\mathbb {E}\left [{X}\right]\mathbb {E}\left [{S}\right] + \mathbb {E}\left [{S^{2}}\right]}{2 \left (\mathbb {E}\left [{X}\right] + \mathbb {E}\left [{S}\right]\right )} + \mathbb {E}\left [{S}\right] \\
=& \frac {\lambda \mu \left (\mathbb {E}\left [{S^{2}}\right] \right )}{2 \left ( {\lambda+\mu} \right)} + \frac{\lambda+\mu}{\lambda \mu},
\end{align}
where (\ref{equ:A3}) is obtained since the inter-arrival time $X$ is exponentially distributed with $\mathbb {E}\left [{X^{2}}\right] = 2\mathbb {E}\left [{X}\right]^{2}$.

%Proof of \emph{Theorem} \ref{thm:1} is then completed.

%\subsection{Proof of Theorem \ref{thm:3}}\label{app:thm3}
%By taking the expectation over inter-departure time $Y$, we have
%\begin{align}
%p_{\rm {mis}}^{\text {M/G/1/1-M}} &= \mathbb {E}\left [{e^{-\nu Y}}\right]\nonumber \\
%&= \mathbb {E}\left [{e^{-\nu \left( X + S \right)}}\right]\nonumber \\
%&= \mathbb {E}\left [{e^{-\nu  X }}\right] \mathbb {E}\left [{e^{-\nu S }}\right]\nonumber \\
%&= G_{X}\left ({-\nu }\right)G_{S}\left ({-\nu }\right).
%\end{align}

%Proof of \emph{Theorem} \ref{thm:3} is then completed.

\subsection{Proof of Theorem \ref{thm:101}}\label{app:thm101}
Firstly, we compare the missing probabilities between the M/U/1/1-M and the M/M/1/1-M bufferless systems with the same Poisson decision process. By setting $t = \nu/\mu > 0$, the difference between them is
\begin{align}
f_{1}\left ({t}\right) &= p_{\rm {mis}}^{\text {M/M/1/1-M}} - p_{\rm {mis}}^{\text {M/U/1/1-M}} \nonumber \\
&= \frac{\lambda}{\lambda+\nu} \left( \frac{1}{1+t} - \frac{1-e^{-2t}}{2t} \right).
\end{align}
Let us denote
\begin{align}
f_{2}\left ({t}\right) &= \frac{2t}{1+t}+e^{-2t}-1,
\end{align}
and its first derivative with respect to $t$
\begin{align}
f_{2}^{\prime}\left ({t}\right) = \frac{2\left( e^{2t}-\left( 1+t \right)^2 \right)}{e^{2t}\left( 1+t \right)^2}.
\end{align}
Let us further denote that $g(t)=e^{2t}-(1+t)^2$ and $g^{\prime}(t)=2(e^{2t}-1-t)$. Noticing that $e^{2t} > 2t+1$, one can check that $g^{\prime}(t) > 0$, which means $g(t)$ monotonically increases with $t$ and $g(t)> g(0) = 0$. Therefore, $f_{2}^{\prime}(t)> 0$ and $f_{2}(t)$ also monotonically increases as $t$ increases with a minimum $f_{2}(t)> f(0) = 0$. Since $G_{X}({-\nu }) = \lambda/(\lambda + \nu)$ is positive and $f_2(t) > 0$, it can already be proved that $f_{1}\left ({t}\right) > 0$ and $p_{\rm {mis}}^{\text {M/U/1/1-M}} < p_{\rm {mis}}^{\text {M/M/1/1-M}}$ for all $t>0$.

Secondly, we make the missing probability comparison between the M/U/1/1-M and the M/D/1/1-M bufferless systems with the same Poisson decision process. Let us denote
\begin{align}
f_{3}\left ({t}\right)= \frac{p_{\rm {mis}}^{\text {M/U/1/1-M}}}{p_{\rm {mis}}^{\text {M/D/1/1-M}}}
= \frac{e^t-e^{-t}}{2t},
\end{align}
and
\begin{align}
f_{3}^{\prime}\left ({t}\right) = \frac{e^t\left( t-1 \right)+e^{-t}\left( t+1 \right)}{2t^2}.
\end{align}
Further, let us denote $h(t) = e^t( t-1 )+e^{-t}( t+1)$ and $h^{\prime}(t)=t(e^t-e^{-t})$. Since $t>0$, it is clear that $h^{\prime}(t)>0$, which means that $h(t)$ increases monotonically with $t$ and $h(t)> h(0) = 0$. Hence, $f_{3}^{\prime} ({t})$ is also positive and $f_{3}(t)> f_{3}(0) = 1$, which proves that $p_{\rm {mis}}^{\text {M/D/1/1-M}} < p_{\rm {mis}}^{\text {M/U/1/1-M}}$ when $t>0$.

%Proof of \emph{Theorem} \ref{thm:101} is then completed.

\subsection{Proof of Theorem \ref{thm:5}}\label{app:thm5}
\begin{figure}[!t]
\centering
\includegraphics[width=1\linewidth]{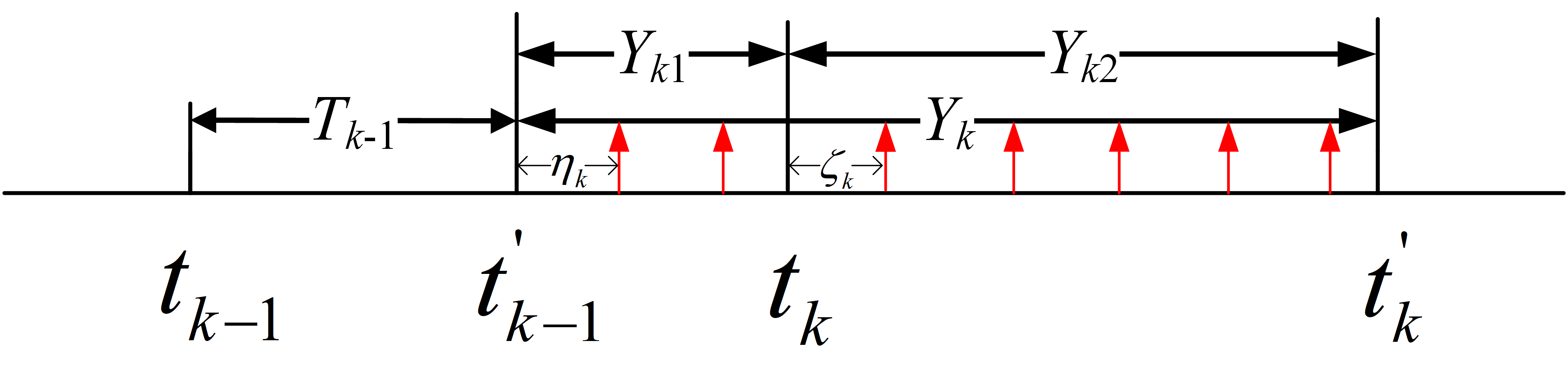}
\caption{System time and inter-departure time for bufferless systems.}
\label{fig:blocking_AuD}
\end{figure}
As shown in Fig. \ref{fig:blocking_AuD}, the time elapsed from the departure moment of the ${(k-1)}^{th}$ successful update to the arrival moment of the $k^{th}$ successful update is denoted as $Y_{k1}$, i.e., $Y_{k1} = t_k-t_{k-1}^{\\'}$. Moreover, we denote the system time of the $k^{th}$ successful update as $Y_{k2}$, i.e., $Y_{k2} = t_{k}^{\\'}-t_{k}$, which is equivalent to its service time $S_k$. In addition, we denote $\eta_k = \tau_{k_1} - t_{k-1}^{\\'}$, which is the period from the departure epoch of the ${(k-1)}^{th}$ successful update to the first decision epoch during $Y_k$. By considering the fact that the decision epochs still follow a uniform distribution in the period $Y_k$, $\eta_k$ is therefore uniformly distributed with parameter $\nu$ over $[0, 1/\nu]$. Thus, the age upon any decision during $Y_k$ is
\begin{align}
\Delta _{\rm D}\left ({\tau _{k_{j}}}\right)=T_{k-1}+\eta_k+\frac {j-1}{\nu },~~j=1, 2,\ldots, N_{k}^{(1)}+N_{k}^{(2)},
\end{align}
in which $N_{k}^{(1)}$ and $N_{k}^{(2)}$ are respectively the numbers of decisions made in the period $Y_{k1}$ and $Y_{k2}$.

Adding up the AuDs of all decisions, the average sum AuD in the period $Y_k$ is given by
\begin{align} \label{equ:D1}
& \mathbb {E}   \left [  { \Delta _{{\rm D}k}}\right] \nonumber \\
&=\mathbb {E}\left [{\Sigma _{j=1}^{N_{k}^{(1)}+N_{k}^{(2)}}\Delta _{\rm D} \left ({\tau _{k_{j}}}\right)}\right] \nonumber \\
 &=\mathbb {E}\left [{T_{k-1}}\right]\left ({\mathbb {E}\left [{N_{k}^{(1)}}\right]+ \mathbb {E}\left [{N_{k}^{(2)}}\right]}\right) \nonumber \\
&\quad +\,\,\frac {\mathbb {E}\left [{\left ({N_{k}^{(1)}}\right)^{2}}\right]+\mathbb {E} \left [{\left ({N_{k}^{(2)}}\right)^{2}}\right]+2{\mathbb {E}\left [{N_{k}^{(1)}}\right]\mathbb {E}\left [{N_{k}^{(2)}}\right]}}{2\nu }.
\end{align}

Suppose there are $K$ successful updates and $N_{T}$ decisions during an infinite period $T$, it has
\begin{align} \label{equ:D2}
\lim _{T\rightarrow \infty }\frac {N_{T}}{K}= \nu \mathbb {E}\left [{Y_{k}}\right].
\end{align}

Therefore, the average AuD in the M/G/1/1-D bufferless system is
\begin{align}%��ʽ̫���ˣ�������
&\overline{\Delta}_{\text {M/G/1/1-D}}^{\rm{blocking}} \nonumber \\
&=\lim _{T\rightarrow \infty }\frac {1}{N_{T}} \sum _{k=1}^{K} \Delta _{{\rm D}k} \nonumber
\\ &=\lim _{T\rightarrow \infty }\frac {K}{N_{T}}\left ({\frac {1}{K} \sum _{k=1}^{K}\Delta _{ {\rm D}k}}\right)\nonumber
%\\ &=\frac{\mathbb {E}\left [{\Delta _{{\rm D}k}}\right]}{\nu \mathbb {E}\left [{Y_{k}}\right]}\nonumber
\\ &=\frac {\mathbb {E}\left [{T_{k-1}}\right]\left (\mathbb {E}\left [{N_{k}^{(1)}}\right]+\mathbb {E}\left [{N_{k}^{(2)}}\right] \right)}{\nu\mathbb {E}\left [{Y_k}\right]}\nonumber
\\& \quad +\frac {\mathbb {E}\left [{\left ({N_{k}^{(1)}}\right)^{2}}\right]+\mathbb {E} \left [{\left ({N_{k}^{(2)}}\right)^{2}}\right]+2{\mathbb {E}\left [{N_{k}^{(1)}}\right] \mathbb {E}\left [{N_{k}^{(2)}}\right]}}{2{\nu^2}\mathbb {E}\left [{Y_{k}}\right]}.
\end{align}

\subsection{Proof of Corollary \ref{cor:3}}\label{app:cor3}
Let us denote $\alpha = e^{-\mu / \nu}$ and $\beta = e^{-\lambda / \nu}$. Considering that $Y_{k1}$ is the residual inter-arrival time of the ${(k-1)}^{th}$ and $k^{th}$ successful update, it has an exponential distribution $f_{{Y}_{k1}}(t)=\lambda e^{-\lambda t }$.

Firstly, we consider the case when the service time follows a uniform distribution. We denote the decision numbers in $Y_{k1}$ and $Y_{k2}$ as $N_{{\rm U}k}^{(1)}$ and $N_{{\rm U}k}^{(2)}$, respectively. It has
\begin{align} \label{equ:MU11D1}
\Pr \left \{{N_{{\rm U}k}^{(1)}=j}\right \}=&\Pr \left \{{\frac {j-1}{\nu }+\eta< Y_{k1}< \frac {j}{\nu }+\eta}\right \} \nonumber
\\=&\int _{0}^{\frac {1}{\nu }}f_{\rm \eta}\left ({x}\right)dx\int _{\frac {j-1}{\nu }+x}^{\frac {j}{\nu }+x} f_{{Y}_{k1}}\left ({t}\right)dt \nonumber\\
=&\frac {\nu \left ({1-\beta}\right)^2 \beta^{j-1}}{\lambda},~~j=1, 2,\ldots \nonumber
\\ \mathbb {E}\left [{N_{{\rm U}k}^{(1)}}\right]=&\frac {\nu }{\lambda },\mathbb {E}\left [{\left ({N_{{\rm U}k}^{(1)}}\right)^{2}}\right] =\frac {\nu \left ({1+\beta}\right)}{\lambda \left ({1-\beta}\right)}.
\end{align}

We also denote $\zeta_k = \tau_{k_{N_{k}^{(1)}+1}} - t_{k}$, which is the period from the arrival epoch of the $k^{th}$ successful update to the first decision epoch during $Y_{k2}$. Similarly, $\eta_k$ is uniformly distributed with parameter $\nu$ over $[0, 1/\nu]$ because the decision epochs follow a uniform distribution during $Y_k$. Hence, it has
\begin{align} \label{equ:MU11D2}
\Pr \left \{{N_{{\rm U}k}^{(2)}=j}\right \}=&\Pr \left \{{\frac {j-1}{\nu }+\zeta< Y_{k2}< \frac {j}{\nu }+\zeta}\right \} \nonumber \\
=&\int _{0}^{\frac {1}{\nu }}f_{\zeta}\left ({x}\right)dx\int _{\frac {j-1}{\nu }+x}^{\frac {j}{\nu }+x} f_{Y_{k2}}\left ({t}\right)dt \nonumber \\
=&\frac {1}{2m_{0}},~~j=1, 2,\ldots \nonumber \\
\mathbb {E}\left [{N_{{\rm U}k}^{(2)}}\right]=&\frac {2m_{0}+1}{2}, \nonumber \\
\;\mathbb {E}\left [{\left ({N_{{\rm U}k}^{(2)}}\right)^{2}}\right]=&\frac {\left ({2m_{0}+1}\right) \left ({2m_{0}+2}\right)\left ({4m_{0}+3}\right)}{12m_{0}}.
\end{align}

Recall that the mean service time is $1/\mu$ and $\mathbb {E}[ Y ] = \mathbb {E}[X]+\mathbb {E}[S]$. By inserting (\ref{equ:MU11D1}) and (\ref{equ:MU11D2}) to (\ref{equ:MG11D}), the average AuD of the M/U/1/1-D bufferless system is

\begin{align} \overline{\Delta}_{\text {M/U/1/1-D}}^{\rm{blocking}}=&\frac {2\rho m_0 + 4m_0 +\rho + 1}{2\mu m_0 \left ({1+\rho }\right) } +\frac {\left ({1+\beta}\right)}{2\mu m_{0}\left ({1+\rho }\right)\left ({1-\beta}\right)} \nonumber\\
&+\frac {\rho \left ({8m_0^3+18m_0^2+13m_0+3}\right)}{12\mu m_{0}^3\left ({1+\rho}\right)}.
\end{align}

Secondly, we consider an M/M/1/1-D system employing a length-1 blocking queue with negative exponential service time distribution. We denote the decision numbers in $Y_{k1}$ and $Y_{k2}$ as $N_{{\rm E}k}^{(1)}$ and $N_{{\rm E}k}^{(2)}$, respectively. Note that $N_{{\rm E}k}^{(1)}$ is equal to $N_{{\rm U}k}^{(1)}$ and all it needs is the first and second moments of $N_{{\rm E}k}^{(2)}$, which is given by
\begin{align} \label{equ:MM11D1}
\Pr \left \{{N_{{\rm E}k}^{(2)}=j}\right \}=&\Pr \left \{{\frac {j-1}{\nu }+\zeta< Y_{k2}< \frac {j}{\nu }+\zeta}\right \} \nonumber \\
=&\int _{0}^{\frac {1}{\nu }}f_{\zeta}\left ({x}\right)dx \int _{\frac {j-1}{\nu }+x}^{\frac {j}{\nu }+x}f_{{Y_{k2}}}\left ({t}\right)dt \nonumber \\%\nonumber
=&\frac{\nu \left ({1-\alpha}\right)^2 \alpha^{j-1}}{\mu},~~j=1, 2,\ldots \nonumber \\
\mathbb {E}\left [{N_{{\rm E}k}^{(2)}}\right]=&\frac {\nu }{\mu }, \mathbb {E}\left [{\left ({N_{{\rm E}k}^{(2)}}\right)^{2}}\right]= \frac {\nu \left ({1+\alpha}\right)}{\mu \left ({1-\alpha}\right)}.
\end{align}

By inserting (\ref{equ:MU11D1}) and (\ref{equ:MM11D1}) to (\ref{equ:MG11D}), the average AuD in the M/M/1/1-D bufferless system with a deterministic decision process is
\begin{align}
\overline{\Delta}_{\text {M/M/1/1-D}}^{\rm{blocking}}=&\frac {2+\rho}{\mu \left ({1+\rho }\right)}+ \frac {\rho \left ({1+\alpha}\right)}{2\mu m_{0}\left ({1+\rho }\right)\left ({1-\alpha}\right)} \nonumber \\ 
&+\frac {\left ({1+\beta}\right)}{2\mu m_{0}\left ({1+\rho }\right)\left ({1-\beta}\right)}.
\end{align}

Finally, when the update service time is deterministic and set to $1/\mu$, the decision numbers in $Y_{k1}$ and $Y_{k2}$ are denoted as $N_{{\rm D}k}^{(1)}$ and $N_{{\rm D}k}^{(2)}$, respectively. It is clear that $N_{{\rm D}k}^{(2)} = m_0$ because the decision intervals are also deterministically distributed with parameter $1/\nu$. Therefore, it has
\begin{align} \label{equ:MD11D1}
\mathbb {E}\left [{N_{{\rm D}k}^{(1)}}\right]=&\mathbb {E}\left [{N_{{\rm U}k}^{(1)}}\right]=\frac {\nu }{\lambda }, \mathbb {E}\left [{\left ({N_{{\rm D}k}^{(1)}}\right)^{2}}\right]=\frac {\nu \left ({1+\beta}\right)}{\lambda \left ({1-\beta}\right)}, \nonumber \\
\mathbb {E}\left [{N_{{\rm D}k}^{(2)}}\right]=&m_0, \mathbb {E}\left [{\left ({N_{{\rm D}k}^{(2)}}\right)^{2}}\right]= m_0^2.
\end{align}

By inserting (\ref{equ:MD11D1}) to (\ref{equ:MG11D}), the average AuD in the M/D/1/1-D bufferless system with a deterministic decision process can be expressed as
\begin{align}
\overline{\Delta}_{\text {M/D/1/1-D}}^{\rm{blocking}}=&\frac {4+3\rho}{2\mu \left ({1+\rho }\right)}+\frac {\left ({1+\beta}\right)}{2\mu m_{0}\left ({1+\rho }\right)\left ({1-\beta}\right)}.
\end{align}

%Proof of \emph{Corollary} \ref{cor:3} is then completed.

%Proof of \emph{Theorem} \ref{thm:7} is then completed.

\subsection{Proof of Theorem \ref{thm:6}}\label{app:thm6}
We first compare the average AuDs between the M/U/1/1-D and the M/M/1/1-D bufferless systems with the same deterministic decision process. Their difference is
\begin{align}
f_{1}\left ({m_{0}}\right)=&\overline{\Delta}_{\text {M/U/1/1-D}}^{\rm{blocking}}- \overline{\Delta}_{\text {M/M/1/1-D}}^{\rm{blocking}}\nonumber\\
=&\frac {\rho \left ({8m_0^3+18m_0^2+13m_0+3}\right)}{12\mu m_{0}^3\left ({1+\rho}\right)}+\frac{1}{2\mu m_0} \nonumber\\
&-\frac {\rho \left ({1+\alpha}\right)}{2\mu m_{0}\left ({1+\rho }\right)\left ({1-\alpha}\right)},
\end{align}
where $\alpha$ has been defined in Appendix \ref{app:cor3}. Set $t = -1/m_0$ and the first derivative of $f_1(t)$ with respect to $t$ is given by
\begin{align} \label{equ:MG11D2}
f_{1}^{\prime }\left ({t}\right)=& f_{1}^{\prime }\left ({-\frac {1}{m_{0}}}\right) \nonumber \\
=&\frac {\rho e^{t}\left (t+1{-e^{t}}\right)}{\mu \left ({1+\rho}\right) \left ({1-e^{t}}\right)^{2}}
+\frac{\rho x \left( 26 - 9t \right)}{12 \mu \left( 1+\rho \right)} \nonumber \\ 
&-\frac {3\rho+1}{2\mu \left (1+\rho\right)}.
\end{align}
Since it has $\rho>0$ and $-1<t<0$, one can observe that each term of (\ref{equ:MG11D2}) is negative and $f_{1}^{\prime }({t}) < 0$, which means $f_{1}({m_{0}})$ is monotonically deceasing with $m_0$. Also note that $f_{1}({0})>0$ and $f_{1}({+\infty})<0$, it can therefore be deduced that there exists a positive integer $m^*_0$ such that $\overline{\Delta}_{\text {M/U/1/1-D}}^{\rm{blocking}}<\overline{\Delta}_{\text {M/M/1/1-D}}^{\rm{blocking}}$ when $m_0> m_0^*$ and $\overline{\Delta}_{\text {M/M/1/1-D}}^{\rm{blocking}}<\overline{\Delta}_{\text {M/U/1/1-D}}^{\rm{blocking}}$ when $m_0 \leq m_0^*$.

Then, we make the average AuD comparison between the M/M/1/1-D and the M/D/1/1-D bufferless systems, their difference is
\begin{align}
f_{2}\left ({m_{0}}\right)=&\overline{\Delta}_{\text {M/M/1/1-D}}^{\rm{blocking}}- \overline{\Delta}_{\text {M/D/1/1-D}}^{\rm{blocking}}\nonumber\\
=&\frac {\rho \left ({1+\alpha}\right)}{2\mu m_{0}\left ({1+\rho }\right)\left ({1-\alpha}\right)}-\frac{\rho}{2\mu \left( 1+\rho \right)}.
\end{align}
Since $(1+\alpha)/(m_0 (1-\alpha))$ decreases as $m_0$ increases and $f_{2}(\infty)>0$, one can check that $\overline{\Delta}_{\text {M/M/1/1-D}}^{\rm{blocking}}>\overline{\Delta}_{\text {M/D/1/1-D}}^{\rm{blocking}}$ for all $m_0 > 1$.

Lastly, we compare the average AuD of M/U/1/1-D queuing model to that of M/D/1/1-D queuing model, the difference is
\begin{align}
f_{3}\left ({m_{0}}\right)=&\overline{\Delta}_{\text {M/U/1/1-D}}^{\rm{blocking}}- \overline{\Delta}_{\text {M/D/1/1-D}}^{\rm{blocking}} \nonumber\\
=&\frac {\rho \left ({8m_0^3+18m_0^2+13m_0+3}\right)}{12\mu m_{0}^3\left ({1+\rho}\right)}+\frac{1}{2\mu m_0} \nonumber \\
&-\frac {\rho}{2\mu \left ({1+\rho }\right)}.
\end{align}
It is clear that $f_{3}({m_{0}})$ drops with $m_0$ and $f_{3}(\infty)>0$. Hence, it can be obtained that $\overline{\Delta}_{\text {M/U/1/1-D}}^{\rm{blocking}}>\overline{\Delta}_{\text {M/D/1/1-D}}^{\rm{blocking}}$ when $m_0>1$.

\subsection{Proof of Theorem \ref{thm:7}}\label{app:thm7}
Firstly, let us prove that the average AuDs in \emph{Corollary} \ref{cor:3} decrease with $m_0$. Let us define
\begin{align}
f\left({m_{0}}\right) = \frac{1+\alpha}{m_0\left( 1-\alpha \right)}.
\end{align}
By setting $t = 1/m_0$, the first derivative of $f(t)$ with respect to $t$ is given by
\begin{align}
f^{\prime }\left({t}\right) = \frac{e^{2t}-2te^t-1}{\left(e^t-1\right)^2}.
\end{align}
Let us further denote $g(t)=e^{2t}-2te^t-1$ and $g^{\prime}(t)=2e^{2t}-2te^t-2e^t$. Noticing that $e^t \geq t+1$, one can check that $g^{\prime}(t)\geq 0$ and $g(t)\geq 0$. Therefore, it is observed that $f^{\prime}(t)$ is non-negative, which means $f({m_{0}})$ is monotonically deceasing with $m_0$. Also note that this trend remains unchanged if $\alpha$ is replaced by $\beta$. Hence, it can be verified that the average AuDs in \emph{Corollary} \ref{cor:3} decrease with $m_0$ since each item in (\ref{equ:MU11D}), (\ref{equ:MM11D}) and (\ref{equ:MD11D}) decreases or stays unchanged with $m_0$.

Moreover, when $m_0 \to \infty$ , one can get $(1+\alpha)/(m_0(1-\alpha)) = 2$ and $(1+\beta)/(m_0(1-\beta))= 2/\rho$. By combining these equations, one can check that the average AuD under deterministic decisions reduces to the corresponding average AuD under Poisson decisions as $m_0 \to \infty$.

%Proof of \emph{Theorem} \ref{thm:6} is then completed.

\begin{figure}[!t]
%\centering
\includegraphics[width=1\linewidth]{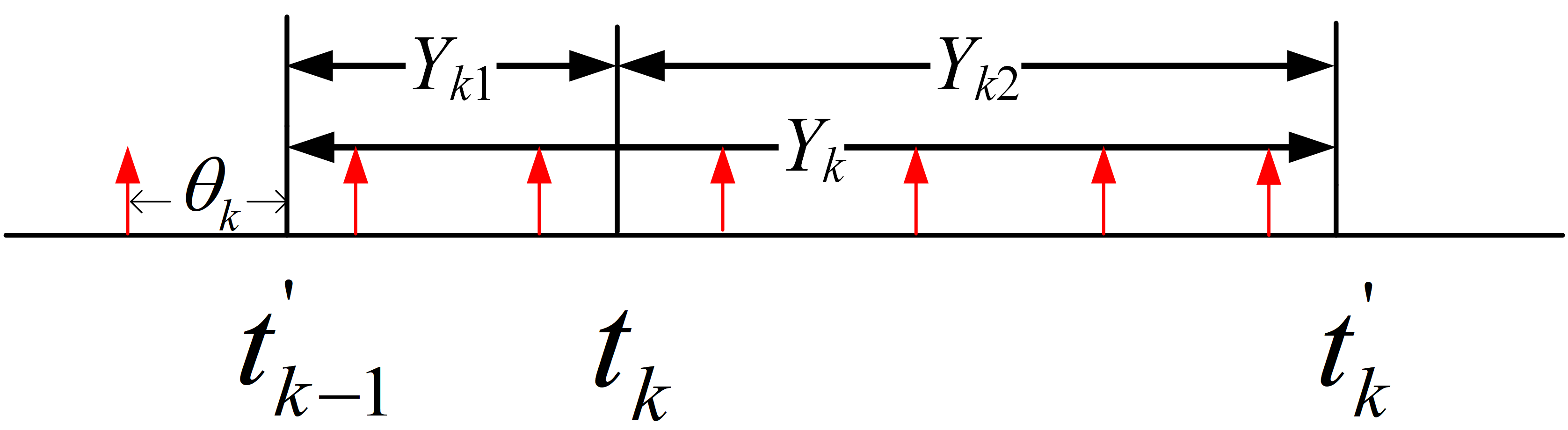}
\caption{Inter-departure time.}
\label{fig:blocking_mis}
\end{figure}

\subsection{Proof of Corollary \ref{cor:4}}\label{app:cor4}
As shown in Fig. \ref{fig:blocking_mis}, we denote $\theta_k = t_{k-1}^{\\'} - \tau_{{k-1}_{N_{k-1}}}$ as the period between the departure epoch of the ${(k-1)}^{th}$ successful update and the latest decision epoch before $t_{k-1}^{\\'}$. In particular, let us make an approximation that $\theta_k$ is uniformly distributed with $1/\nu$ whose PDF is $f_{\theta}\left ({t}\right)={\nu }$ for $t\in \left (0,1/\nu \right)$. Since the missing probability is equal to the probability that decision interval is larger than inter-departure time, by considering the fact that the decision interval is deterministic and equal to $1/\nu$, it has
\begin{align}
& p_{\rm {mis}}^{\text {M/G/1/1-D}} \nonumber \\
&=\Pr \left \{{Y_{k}+\theta< \frac {1}{\nu }}\right \} \nonumber \\
&=\Pr \left \{{Y_{k1}+Y_{k2}+\theta< \frac {1}{\nu }}\right \} \nonumber \\
&=\int _{0}^{\frac {1}{\nu }}f_{S}\left ({x}\right)dx\int _{0}^{\frac {1}{\nu }-x}f_{{Y}_{k1}} \left ({y}\right)dy\int _{0}^{\frac {1}{\nu }-x-y}f_{\theta}\left ({z}\right)dz \nonumber \\
&=\int _{0}^{\frac {1}{\nu }}f_{S}\left ({x}\right)\left ({\frac {\nu \beta e^{\lambda x}}{\lambda }-\nu x+\frac {\lambda-\nu }{\lambda }}\right)dx, \\
&p_{\rm {mis}}^{\text {M/U/1/1-D}}=\frac{1}{4m_0}+\frac {{m_{0} \left ({1-\beta}\right)-\rho }}{2\rho ^{2}}, \\
&p_{\rm {mis}}^{\text {M/M/1/1-D}}=\frac{m_0 \left ( \alpha- \beta \right )}{\rho \left( \rho-1 \right)}
+\frac{\left ( \rho-m_0- \rho m_0 \right ) \left (1-\alpha \right )}{\rho}+{\alpha},
\end{align}
in which $\alpha$ and $\beta$ have been defined in Appendix \ref{app:cor3}.

In the deterministic service time case, however, it is clear that each inter-departure time consists at least $m_0$ decisions, which means all the successful update will be used at least $m_0$ times. Since $m_0 \geq 1$, it has
\begin{align}
p_{\rm {mis}}^{\text {M/D/1/1-D}} = 0.
\end{align}

%Proof of \emph{Corollary} \ref{cor:4} is then completed.

\baselineskip=18pt
\bibliographystyle{IEEEtran}
\bibliography{IEEEabrv,AuD}

\ifCLASSOPTIONcaptionsoff
  \newpage
\fi
\end{document}